\documentclass[12pt,reqno]{amsart}
\usepackage[top=2.54cm, bottom=2.54cm, left=3.18cm, right=3.18cm]{
geometry}
\usepackage{dsfont, amssymb,amsmath,amscd,latexsym, amsthm, amsxtra,amsfonts}
\usepackage[all]{xy}
\usepackage[active]{srcltx}

\usepackage[round]{natbib}
\usepackage{bbm}
\usepackage{enumerate}
\usepackage{mathrsfs}
\usepackage{graphicx}
\usepackage{comment}
\usepackage{mathtools}

\usepackage{tikz}
\usetikzlibrary{calc,arrows}
\usepackage{verbatim}
\usepackage{graphicx}
\usepackage{subfigure}

\newtheorem{theorem}{Theorem}[section]
\newtheorem{corollary}[theorem]{Corollary}
\newtheorem{proposition}[theorem]{Proposition}
\newtheorem{lemma}[theorem]{Lemma}

\theoremstyle{definition}

\renewcommand{\theequation}{\arabic{section}.\arabic{equation}}
\newtheorem{Def}{Definition}[section]

\theoremstyle{definition}

\theoremstyle{definition}
\newtheorem{remark}{Remark}
\theoremstyle{definition}

\newcommand{\rd}{\mathrm{d}}

\renewcommand{\epsilon}{\varepsilon}

\renewcommand{\cite}{\citet*}
\usepackage[pdfstartview=FitH, bookmarksnumbered=true,bookmarksopen=true, colorlinks=true, pdfborder=001, citecolor=blue, linkcolor=blue,urlcolor=blue]{hyperref}
\usepackage{graphics}
\graphicspath{{figures/}}
\usepackage[round]{natbib}
\setcitestyle{authoryear,round}

\begin{document}
\makeatletter
\def\@setauthors{%
\begingroup
\def\thanks{\protect\thanks@warning}%
\trivlist \centering\footnotesize \@topsep30\p@\relax
\advance\@topsep by -\baselineskip
\item\relax
\author@andify\authors
\def\\{\protect\linebreak}%
{\authors}%
\ifx\@empty\contribs \else ,\penalty-3 \space \@setcontribs
\@closetoccontribs \fi
\endtrivlist
\endgroup } \makeatother
 \baselineskip 18pt
 \title[{{\tiny Optimal management of DC pension fund}}]
 {{\tiny
Optimal management of DC pension fund under relative performance ratio and VaR constraint}} \vskip 10pt\noindent
\author[{\tiny  Guohui Guan, Zongxia Liang, Yi xia}]
{\tiny {\tiny  Guohui Guan$^{a,\dag}$, Zongxia Liang$^{b,\ddag}$, Yi Xia$^{b,*}$}
 \vskip 10pt\noindent
{\tiny ${}^a$School of Statistics, Renmin University of China, Beijing 100872, China
\vskip 10pt\noindent\tiny ${}^b$Department of Mathematical Sciences, Tsinghua
University, Beijing 100084, China
}
  \footnote{\\
 $^{\dag}$\ e-mail: guangh@ruc.edu.cn\\
 $^{\ddag}$\ Corresponding author, e-mail:  liangzongxia@mail.tsinghua.edu.cn\\
 $^*$ Corresponding author,\ e-mail:  xia-y20@mails.tsinghua.edu.cn
  }}
\numberwithin{equation}{section}
\maketitle
\noindent
\begin{abstract}
In this paper, we investigate the optimal management of defined contribution (abbr. DC) pension plan under relative performance ratio and Value-at-Risk (abbr. VaR) constraint. Inflation risk is introduced in this paper and the financial market consists of cash, inflation-indexed zero coupon bond and a stock. The goal of the pension manager is to maximize the performance ratio of the real terminal wealth under VaR constraint. An auxiliary process is introduced to transform the original problem into a self-financing problem first. Combining linearization method, Lagrange dual method, martingale method and concavification method, we obtain the optimal terminal wealth under different cases. For convex penalty function, there are fourteen cases while for concave penalty function, there are six cases. Besides, when the penalty function and reward function are both power functions, the explicit  forms of the optimal investment strategies are obtained. Numerical examples are shown in the end of this paper to illustrate the impacts of the performance ratio and VaR constraint.

\vskip 5pt  \noindent
Submission Classification: IB13, IB81,  IE43, IE22, IE13, IE53.
\end{abstract}
\vskip15pt
\setcounter{equation}{0}
\section{\bf Introduction}
The management of pension fund is a popular topic in actuarial science. Pension plans serve as a means of financial stability and security after retirement. For an individual, joining in a pension plan is an efficient way to ensure a comfortable retired life. Generally, there are two kinds of pension funds: defined contribution (abbr. DC) pension fund and defined benefit (abbr. DB) pension fund. In the DC pension plan, the contribution rate is fixed in advance and the benefits after retirement depend on the wealth of the fund at retirement. As the risk in the DC pension plan is undertaken by the pension participators, DC pension plan accounts for more and more percentage in the retirement system.
\vskip 5pt
In the DC pension plan, the investment during the accumulation phase influences the benefits after retirement and attracts a lot of attentions. There are many literature concerning the optimal management of DC pension fund during the accumulation phase. The optimal investment strategies maximizing the  Constant Relative Risk Aversion (abbr. CRRA) utility function of terminal wealth are first derived in \cite{boulier2001optimal}. As the time horizon of the pension fund is relatively long, the pension manager is faced with various risks. The interest rate influences the expected returns of the assets in the financial market and then affects the fund wealth heavily. The interest risk for DC pension fund is studied in \cite{CBD}, \cite{zhang2013optimal}, \cite{guan2014optimal}, \cite{njoku2017effect}, etc. Another important financial risk is the inflation risk, which decreases the real purchasing power of the wealth, see  \cite{battocchio2004optimal}, \cite{han2012optimal}, \cite{ZKE}, \cite{ZE}, etc. Besides, empirical studies show that the equity price does not follow the geometric Brownian motion in general. The Constant Elasticity of Variance (abbr. CEV) model is studied in \cite{gao2009optimal}, \cite{he2020optimal}, etc. \cite{guan2014optimal} and \cite{ma2020optimal} consider the Heston's stochastic volatility model for a DC pension fund. \cite{sun2016precommitment} and \cite{mudzimbabwe2019simple} study the jump diffusion process in the pension fund. All the above work aim to present a comprehensive view  about the financial market and help manage the pension fund better.
\vskip 5pt
Besides the characterization of the financial market, the optimization rule also affects the investment strategies of the DC pension fund and is remarkable. However, most of the previous mentioned work concern maximizing the expected utility of terminal wealth. When considering the CRRA or  Constant Absolute Risk Aversion (abbr. CARA) utility function of terminal wealth, stochastic dynamic programming method or martingale method (in a complete market) can be efficiently applied. In the case of CRRA utility function, ignoring the contribution rate, the optimal allocations are proportional to the fund wealth, see \cite{vigna2009mean}, \cite{guan2014optimal}. The optimal investment allocation does not depend on the fund wealth in the case of CARA utility function, see \cite{gao2009optimal}. The main shortcoming of the above optimization rules is that they do not distinguish the preferences for gains and losses. However, as shown in \cite{tversky1974judgment}, the individual acts differently towards gains and losses. \cite{guan2016optimal} employ the S-shaped utility in DC pension fund, and the optimal strategies have a complicate relationship with the fund wealth. Moreover, Omega ratio proposed in \cite{OmegaRatio} is also an efficient tool to distinguish gains and losses, which provides a new perspective on performance assessment of the fund. However, the Omega ratio is defined as the ratio between two expectations, and is non-linear. \cite{2019OptimalB} shows that in the continuous setting, the maximization problem of the Omega ratio is unbounded  and then adds additional constraint in the problem. In \cite{POWPR}, the Omega ratio is modified to contain the reward and penalty functions and the optimization problem is well-posedness. In this work, we concentrate on the relative performance ratio as in \cite{POWPR}. We are interested in the effects of the performance ratio on the economic behaviors of the pension manager.
\vskip 5pt
The pension manager often expects to achieve higher optimization goal. However, in reality, to ensure the stability of the insurance industry, the behaviours of the pension manager should satisfy some regulatory requirements.  Recently, the Solvency II introduces a harmonised, sound and robust prudential framework for the regulatory of  insurance company in Europe. Identifying the risk and calculating the solvency capital requirement is important in Solvency II. A pension plan often lasts for 20-40 years and inflation risk is a very prominent factor influencing the purchasing power of the fund. As such, it is necessary to introduce inflation risk in a pension fund. Although the risk of the DC pension fund is  undertaken by the participators, maintaining a certain solvency ability ensures the stability of the fund and can attract more pension participators. As such, we introduce the VaR constraint for the pension fund at retirement. In the DC management, the VaR constraint has been studied in \cite{guan2016optimal}, \cite{OIWS}. The optimization rule in their work is the expected utility form and the optimal  strategies can be derived by martingale method and dual method.
\vskip 5pt
In this paper, we study the optimal management of DC pension fund under inflation risk. The financial market consists of cash, inflation indexed zero coupon bond and stock. The pension manager receives a continuously stochastic contribution rate from pension participators. On the one hand, the optimization rule in our work is not the expected utility maximization problem while characterized by the performance ratio as in \cite{POWPR}. On the other hand, the manager expects that the terminal fund wealth is above a certain level and has a VaR constraint. We investigate the optimization problem and obtain the optimal terminal wealth as well as the optimal investment strategies in some cases. Different from the work of \cite{POWPR} without VaR constraint, we will see that VaR constraint involves in an additional Lagrange multiplier and the existence of two Lagrange multipliers is not easy to be shown. Besides, the optimal payoff in \cite{POWPR} has one (two) cases for concave (convex) penalty function. The optimal payoff in our paper is more complicated. The financial model in our work contains inflation risk and is non-self-financing, which is also more complex than \cite{POWPR}. Another related work is \cite{OIWS}, which considers S-shaped utility with VaR constraint. The  linearized problem in our paper contains the  S-shaped utility as a special case. However, we also investigate the case of convex penalty function and has a piece-wise concave utility function in the linearized problem. The case of convex penalty function is more complicate and contains more cases, which is interesting.
\vskip 5pt
We have the following contributions in this paper. First, we present the optimization problem under performance ratio and VaR constraint. The optimization goal with performance ratio is not the standard expected utility form and is non-linear. The VaR constraint ensures that the DC fund has a sufficient wealth at retirement and helps manage the fund better. However, combining the non-linear goal with the VaR constraint, the optimization problem becomes very complicated. Second, we introduce inflation risk in our paper and consider a stochastic contribution rate. In order to hedge the inflation risk, an inflation-indexed zero coupon bond is introduced. The wealth process of the fund becomes non-self-financing and we introduce an auxiliary process to transform it into a self-financing one. Third, combining fractional programming, Lagrange dual method, martingale method, concavification method, we obtain the optimal terminal wealth under different cases. When the penalty function is convex (concave), there are fourteen (six) cases. Moreover, when the penalty and reward functions are both linear, the explicit forms of the optimal  investment strategies are obtained.  Fourth, in the Lagrange dual problem, there are two Lagrange multipliers. When there is only one Lagrange multiplier, the existence can be easily revealed. In the case of two Lagrange multipliers, the existence may not hold at the same time and we discuss it in detail. Last, numerical results  are shown to illustrate the impacts of the performance ratio and VaR constraint on the pension manager.
\vskip 5pt
The rest of this paper is organized as follows. Section 2 presents the financial model. The optimization problem under performance ratio with VaR constraint is introduced in Section 3. Section 4 solves the optimization problem. The optimal investment strategies under some specific cases are presented in Section 5. Section 6 shows the numerical results and Section 7 is a conclusion.
\vskip 15pt
\section{\bf Financial model}
Let $ (\Omega,\mathcal{F},\mathbb{F},\mathbb{P}) $ be a complete filtered  probability space with filtration  $ \mathbb{F}\!= \!\{\mathcal{F}_t|0\!\leqslant\! t\!\leqslant\! T\} $ and $ \mathcal{F}_t $ is the information available before time $ t $ in the market. The pension fund starts at initial time 0 and the retirement time is $ T $. The pension fund manager can adjust the strategy within time horizon $[0,T]$. All the processes introduced below are assumed to be well-defined and adapted to $\mathbb{F}$. We suppose that there are no transaction costs and short selling is allowed.
\vskip 10pt
\subsection{\bf Financial market}
In this section, we present the financial market in which the pension fund manager allocates the wealth. Because the investment horizon of a DC pension fund lasts long, often 20-40 years, the wealth of the fund is faced with various risks from the financial market. Especially, inflation risk will decrease the real purchasing power of the fund and plays  a prominent part in the risk management. In our work, we consider the effect of inflation risk on the DC pension fund.  We introduce the financial market with inflation risk and three different assets. To simplify the financial model, we assume that the nominal interest rate and the real interest rate are deterministic. We are mainly concerned with the impact of inflation on the behaviour of the pension fund.  The financial market in our model consists of cash, bond and stock.
\vskip 5pt
The price of the risk-free (i.e., cash) asset $S_0=\{S_0(t)|0\!\leqslant\! t\!\leqslant\! T\}$ is characterized by
\begin{equation}\label{equ-S0}
	\begin{split}
		\frac{\mathrm{d}S_0(t)}{S_0(t)}=r_n(t)\mathrm{d}t,\quad \quad S_0(0)=s_0,
	\end{split}
\end{equation}
where $s_0$ is a constant. $r_n=\{r_n(t)|0\!\leqslant\! t\!\leqslant\! T\}$ is deterministic and represents nominal interest rate in the market.
\vskip 5pt
Next, we present the financial model of the inflation index by the Fisher equation. The Fisher equation describes relationships among the nominal interest rate $r_n$, the real interest rate $r_r=\{r_r(t)|0\!\leqslant\! t\!\leqslant\! T\}$ and the inflation index $I=\{I(t)|0\!\leqslant\! t\!\leqslant\! T\}$. The inflation index $I$ reflects a reduction in purchasing power per unit of money.
\vskip 5pt
We present the following extended continuous-time Fisher equation given by (cf. \cite{ZKE}, \cite{GL}):
\begin{equation*}
	\left\{
	\begin{split}
		&r_n(t)-r_r(t)=\lim_{\Delta t\rightarrow 0^+}\frac{1}{\Delta t}\widetilde{\mathbb{E}}[i(t,t+\Delta t)|\mathcal{F}_t],\\
		&i(t,t+\Delta t)=\frac{I(t+\Delta t)-I(t)}{I(t)},
	\end{split}
	\right.
\end{equation*}
where $\widetilde{\mathbb{E}}$ is the expectation under risk neutral measure $\widetilde{\mathbb{P}}$ and $i(t,t+\Delta t)$ is  the inflation rate from time $t$ to $t+\Delta t$.  Thus, we characterize the risk of the inflation by the Brownian motion $\widetilde{W}_I=\{\widetilde{W}_I(t)|0\!\leqslant\! t\!\leqslant\! T\}$,  and the following model is an efficient model of $I$ to satisfy the extended Fisher equation:
\begin{equation*}
	\begin{split}
		\frac{\mathrm{d}I(t)}{I(t)}=(r_n(t)-r_r(t))\mathrm{d}t+\sigma_I\mathrm{d}\widetilde{W}_I(t),
	\end{split}
\end{equation*}
where $\widetilde{W}_I$ is a standard Brownian motion under the risk-neutral measure  $\widetilde{\mathbb{P}}$. 
\vskip 5pt
Denote the inflation risk under the original probability measure $\mathbb{P}$ by $W_I$ and the related market price of risk by $\lambda_I$. Then, based on Girsanov's theorem, we can derive the stochastic inflation index $I$ w.r.t. the original probability measure $\mathbb{P}$ as follows:
\begin{equation}\label{equ-I}
	\begin{split}
		\frac{\mathrm{d}I(t)}{I(t)}\!=\!(r_n(t)\!-\!r_r(t))\mathrm{d}t\!+\!\sigma_I[\lambda_I\mathrm{d}t\!+\!\mathrm{d}W_I(t)],\quad I(0)=i_0.
	\end{split}
\end{equation}
\vskip 5pt
In order to hedge the risk of inflation, we introduce here an inflation-indexed zero coupon bond.  An  inflation-indexed zero coupon bond $P(t,T)$ is a contract at time $ t$ with final payment of  real money \$$1$ at maturity $T$. Different from the general zero-coupon bond, $P(t,T)$ delivers $I(T)$ at maturity $T$. Based on the pricing formula of derivatives, the price of $P(t,T)$  at time $t$ is $P(t,T)=\widetilde{\mathbb{E}}[\exp(-\int_t^Tr_n(s)\mathrm{d}s)I(T)|\mathcal{F}_t]$. As the nominal interest rate in our model is deterministic, a  simple calculation shows that the explicit form of $P(t,T)$ is
\begin{equation*}
	P(t,T)=I(t)\exp[-\int_t^Tr_r(s)\mathrm{d}s].
\end{equation*}
We  see that $P(t,T)$ is in fact the nominal wealth of the discounted wealth of \$$1$ in the real market. $P(t,T)$ also satisfies the following backward stochastic differential equation:
\begin{equation}\label{equ-P}
	\begin{cases}
		\frac{\mathrm{d}P(t,T)}{P(t,T)}=r_n(t)\mathrm{d}t+\sigma_{I}[\lambda_I\mathrm{d}t+\mathrm{d}W_I(t)],\\
		P(T,T)=I(T).
	\end{cases}
\end{equation}
The third asset is a stock in the market. The price $S_1=\{S_1(t)|0\!\leqslant\! t\!\leqslant\! T\}$ of the stock is as follows:
\begin{equation}\label{equ-S}
	\begin{cases}
		\frac{\mathrm{d}S_1(t)}{S_1(t)}=r_n(t)\mathrm{d}t+\sigma_{S_1}(\lambda_I\mathrm{d}t+\mathrm{d}W_I(t))
		+\sigma_{S_2}(\lambda_S\mathrm{d}t+\mathrm{d}W_S(t)),\\
		S_1(0)=s_1,
	\end{cases}
\end{equation}
where $\sigma_{S_1}$ and $\sigma_{S_2}$ are positive constants and represent the volatilities of the stock. $W_S=\{W_S(t)|0\!\leqslant\! t\!\leqslant\! T\}$ is a standard Brownian motion on the probability space $(\Omega, \mathcal{F}, \mathbb{F}, \mathbb{P}     )$ and independent of $W_I$. Moreover, $\lambda_S$ is the market price of risk of $W_S$.

\vskip 10pt
\subsection{\bf DC pension fund}
Before retirement, the pension participants put part of the salary in the DC pension fund. As the salary is often not deterministic while depends on the macroeconomics, it is realistic to assume that the contribution rate is a stochastic process. Denote the contribution rate at time $t$ by $c(t)$. We suppose that $c=\{c(t)|0\!\leqslant\! t\!\leqslant\! T\}$ satisfies the following equation:
\begin{equation*}
	\frac{\rd c(t)}{c(t)}\!=\!\mu \mathrm{d}t\!+\!\sigma_{C_1}\mathrm{d}W_I(t)\!+\!\sigma_{C_2}\mathrm{d}W_S(t),~ c(0)\!=\!c_0,
\end{equation*}
where $\mu>0$ is a constant, $\sigma_{C_1}\geqslant 0$ and $ \sigma_{C_2}\geqslant0$ are the volatilities of the contribution rate. In order to hedge the risks from the financial market, the pension manager invests in the market continuously within time horizon $[0,T]$. Denote money invested in the cash, inflation-indexed zero coupon bond and stock  at time $t$ by $\pi_0(t)$, $ \pi_P(t)$ and $ \pi_S(t)$, respectively. Denote $ \pi_P\triangleq\{\pi_P(t)|0\leqslant t\leqslant T\}   $, $\pi_S\triangleq\{\pi_S(t)|0\leqslant t\leqslant T\}   $. Then, under the investment strategy $\pi\triangleq(\pi_P, \pi_S   )$, the wealth of the DC pension fund is as follows:
\begin{equation}\label{equ-initialX}
	\begin{cases}
		\mathrm{d}X^\pi(t)=c(t)\rd t+\pi_0(t)\frac{\mathrm{d}S_0(t)}{S_0(t)}
		+\pi_P(t)\frac{\mathrm{d}P(t,T)}{P(t,T)}
		+\pi_S(t)\frac{\mathrm{d}S_1(t)}{S_1(t)},\\
		X^\pi(0)=x_0,
	\end{cases}
\end{equation}
where $x_0\geqslant 0$ is the initial wealth of  the pension fund. Substituting Eqs.~(\ref{equ-S0}), (\ref{equ-P}) and (\ref{equ-S}) into Eq.~(\ref{equ-initialX}), we obtain the compact form of the wealth process $ X^\pi=\{X^\pi(t)|0\leqslant t\leqslant T\}  $ as follows:
\begin{equation}\label{equ-X}
	\begin{cases}
		\mathrm{d}X^\pi(t)=\!&c(t)\rd t+r_n(t)X(t)\mathrm{d}t+\pi_P(t)\sigma_I[\lambda_I\mathrm{d}t+\mathrm{d}W_I(t)]
		\\&+\pi_S(t)\sigma_{S_1}[\lambda_I\mathrm{d}t
		+\mathrm{d}W_I(t)]+\pi_S(t)\sigma_{S_2}[\lambda_S\mathrm{d}t+\mathrm{d}W_S(t)],\\
		X^\pi(0)=&x_0,
	\end{cases}
\end{equation}
where $X^\pi(t)=\pi_0(t)+\pi_P(t)+\pi_S(t)$.  We call $\pi=(\pi_P,\pi_S)$  an admissible strategy if it satisfies the following conditions:\begin{enumerate}
	\item [(i)]$\pi_P\text{~and~} \pi_S$ are progressively measurable w.r.t. the filtration $\mathbb{F}$,
	\item [(ii)]$\mathbb{E}\{\int_0^T[\pi_P(t)^2\sigma_I^2+\pi_S(t)^2\sigma_{S_1}^2+\pi_S(t)^2\sigma_{S_2}^2]\mathrm{d}t\}<+\infty$,
	\item [(iii)]Eq.~(\ref{equ-X}) has a unique solution for the  initial data $(t_0,i_0,x_0)\in [0,T]\times [0,+\infty)^2$.
\end{enumerate}
\vskip 5pt
Denote the set of  all admissible investment strategies of $\pi$ by  $\mathcal{A}(x_0)$. We are only concerned with the admissible strategies. The pension fund manager searches the optimal strategy within the admissible set under some optimization criterion.
\vskip 15pt
\section{\bf Optimization problem}
In this section, we present the optimization problem for the pension fund manager under performance ratio and VaR constraint.
\vskip 10pt
\subsection{\bf Performance ratio}
There are many performance measures to evaluate the investment of the fund. The mean-variance criterion in \cite{Markowitz1952Portfolio} has been widely applied in portfolio selection. However, the variance cannot distinguish investors' gains and losses. Besides, the returns are often supposed to be normally distributed in the mean-variance analysis, which often conflicts with the empirical results over the real data. Another widely applied measure is the Sharpe ratio, which employs variance and can not distinguish gains and losses.
\vskip 5pt
\cite{OmegaRatio} introduce the Omega ratio to evaluate the investment performance: given a reference level value $ \theta $, the Omega ratio of the random variable $ R $ is defined by
\begin{equation}\label{equ:ome}
	\Omega_{\theta}(R)=\frac{\mathbb{E}[(R-{\theta})_+]}{\mathbb{E}[({\theta}-R)_+]},
\end{equation}
where $ x_+\triangleq \max\{x,0\} $ represents the positive part of $ x $. $\theta$ is a given reference level. The Omega ratio captures all of the higher moment information in the distribution and  incorporates sensitivity to return levels. Compared with the mean-variance criterion and Sharpe ratio, the Omega ratio distinguishes the gains over $\theta$ and losses below $\theta$.  It also provides a risk-reward evaluation of the returns distribution which incorporates the beneficial impact of gains as well as the detrimental effect of losses, relative to any individual's loss threshold.
\vskip 5pt
As inflation exists in the financial market, the pension manager is concerned with the real terminal wealth at retirement. Therefore, we use the Omega ratio of the real terminal wealth to rate the investment performance of the pension fund. The Omega ratio w.r.t. the real terminal wealth is as follows:
\begin{equation}\label{equ:ome1}
\Omega_{\theta}(\frac{X^{\pi}(T)}{I(T)})\triangleq \frac{\mathbb{E}[(\frac{X^{\pi}(T)}{I(T)}-{\theta})_+]}
		{\mathbb{E}[({\theta}-\frac{X^{\pi}(T)}{I(T)})_+]}.
\end{equation}
We should note here that the Omega ratio of Eq.~(\ref{equ:ome}) in \cite{OmegaRatio} is related with the return of the portfolio. In the case of the simple returns, the Omega ratio of Eq.~(\ref{equ:ome1}) described by the terminal wealth is equivalent to the description by the simple returns.
\vskip 5pt
However, as shown in \cite{2019OptimalB} and \cite{POWPR}, the maximization problem of the Omega ratio is an ill-posed problem. We also show the ill-posedness in Remark \ref{Nonconvex}.  \cite{2019OptimalB} add additional constraints to make the problem well-posedness. In this work, we follow the framework in  \cite{POWPR} to make the optimization problem bounded. We introduce  two weighting functions $ U:\mathbb{R}_+\mapsto \mathbb{R}$ and $ D:\mathbb{R}_+\mapsto \mathbb{R}$ over gains and losses, respectively. $U$ and $D$ are both monotonically increasing measurable. Then the performance ratio in our work is
\begin{equation}\label{equ:ome2}
	R(X^{\pi}(T))\triangleq \frac{\mathbb{E}\left[U\left(\left(\frac{X^{\pi}(T)}{I(T)}-{\theta}\right)_+\right)\right]}{\mathbb{E}\left[D\left(\left({\theta}-\frac{X^{\pi}(T)}{I(T)}\right)_+\right)\right]},
\end{equation}
where the numerator $\mathbb{E}\left[U\left(\left(\frac{X^{\pi}(T)}{I(T)}-{\theta}\right)_+\right)\right]$ represents the income when the real wealth at retirement exceeds the reference level,  the denominator $\mathbb{E}\left[D\left(\left({\theta}-\frac{X^{\pi}(T)}{I(T)}\right)_+\right)\right]$ represents the penalty when the real wealth at retirement is lower than the reference level. Here, the function $U$ is called the reward function, and the function $D$ is called the penalty function.
\vskip 10pt
\subsection{\bf Value-at-Risk}
In this subsection, we formulate the concept of VaR management for the pension manager. VaR is an extremely popular risk measure and many financial companies have successfully used it to manage their risk, particularly for European Bank and Insurance industry, see \cite{dos2010solvency}. VaR is often used to calculate the solvency capital requirement in risk management. For the pension manager, it is also important to introduce VaR in risk management. In particular, for the DB pension fund, if the terminal wealth is lower than some level with large probability, the manager is faced with great insolvency risk. For the DC pension fund, the risk is undertaken by the participants, as such, introducing VaR in the risk management can make the strategy more stable and attract more pension participants.
\vskip 5pt
Now we introduce the VaR for the real terminal wealth $\frac{X(T)}{I(T)}$:
\begin{Def}
	For $\varepsilon \in[0,1]$, the VaR for the real terminal wealth $\frac{X^{\pi}(T)}{I(T)}$ under probability measure $\mathbb{P}$ is
	\begin{equation*}
		\text{VaR}^{\varepsilon}\left(\frac{X^{\pi}(T)}{I(T)}\right)\triangleq-\inf\left\{x\in \mathbb{R}\mid\mathbb{P}\left(\frac{X^{\pi}(T)}{I(T)}\leqslant x\right)>\epsilon\right\}.
	\end{equation*}
\end{Def}
The pension manager uses VaR to measure the risk of the portfolio and require that the VaR of the portfolio can not exceed some level $L\in\mathbb{R}$:
\begin{equation}\label{var-cons}
\operatorname{VaR}^{\varepsilon}\left(\frac{X^{\pi}(T)}{I(T)}\right) \leqslant-L.
\end{equation}
\vskip 5pt
A simple calculation  shows that the VaR constraint (\ref{var-cons}) is equivalent to the following constraint:
\begin{align*}
	\mathbb{P}\left(\frac{X^{\pi}(T)}{I(T)}\geqslant L\right)\geqslant 1-\epsilon, \ \epsilon\in [0,1].
\end{align*}
The VaR constraint requires that the probability of the terminal wealth higher than $L$ is not less than  $1-\epsilon$. Especially, the constraint is not binding for $\epsilon=1$; the VaR constraint requires that the terminal wealth should always be higher than the level $L$ when $\varepsilon=0$, which has been studied in many literature.
\vskip 5pt
\subsection{\bf Optimization rule}
On the one hand, the pension manager is concerned with the modified Omega ratio given by Eq.~(\ref{equ:ome2}) of the real terminal wealth. On the other hand, the manager requires that the VaR of the portfolio does not exceed some given level. As such, the optimization problem of the pension manager is as follows:
\begin{equation}\label{PSV}
	\left\{\begin{aligned}
		\max_{\pi \in \mathcal{A}(x_0)}\quad&\frac{\mathbb{E}\left[U\left(\left(\frac{X^{\pi}(T)}{I(T)}-{\theta}\right)_+\right)\right]}{\mathbb{E}\left[D\left(\left({\theta}-\frac{X^{\pi}(T)}{I(T)}\right)_+\right)\right]},\\
		\text{s.t.}\quad&{X^{\pi}(t)}\text{ satisfies Eq.~(\ref{equ-X}),}\\&{I(t)}\text{ satisfies Eq.~(\ref{equ-I}),}\\
		&\mathbb{P}(\frac{X^{\pi}(T)}{I(T)}\geqslant L)\geqslant1-\varepsilon.
	\end{aligned}\right.
\end{equation}
We see that when $ \varepsilon=1 $, the constraint is equivalent to no constraint, and when $ \varepsilon=0 $, the wealth is required to be higher than a certain level at retirement, which is the optimization problem considered by \cite{Basak}.
\vskip 15pt
\section{\bf Optimal solution}\label{4}
As Eq.~(\ref{equ-X}) evolves a continuously cash flow, the state process $X^\pi$ is not self-financing. Problem (\ref{PSV}) is a complicated optimization problem with non-linear goal, VaR constraint and non-self-financing state process. Problem (\ref{PSV}) can not be directly solved. In order to transform the original problem into a self-financing problem, we first introduce an auxiliary process. Then for the self-financing problem with non-linear goal and VaR constraint, we apply the fractional programming to transform it into a problem with expected utility over terminal wealth under VaR constraint. However, after the linearization, the utility of the terminal wealth may not be globally concave. We employ the concavification method and Lagrange dual method to disentangle the equivalent problem.
\vskip 10pt
\subsection{Auxiliary process}
Because of the existence of contribution rate, the wealth process $ X^{\pi} $ is not a self-financing process. For the convenience of processing, we introduce an auxiliary wealth process $ \tilde{X}^{\tilde{\pi}}\triangleq \{ \tilde{X}^{\tilde{\pi}}(t)|0\leqslant t\leqslant T\} $ to convert the model into an equivalent self-financing one.  First, we introduce the value $D(t,s)$ of a derivative at time $t$ with final payment $c(s)$ at maturity $s$. Based on the theory of derivative pricing, \[D(t,s)=\widetilde{\mathbb{E}}[\exp(-\int_t^sr_n(u)\rd u)c(s)|\mathcal{F}_t].\]
By simple calculation, we obtain the explicit form of $D(t,s)$ as follows:
\begin{equation}\nonumber
	D(t,s)=c(t)\exp[(\mu-\sigma_{c_1}\lambda_I-\sigma_{c_2}\lambda_S)(s-t)-\int_{t}^{s}r_n(u)\rd u].
\end{equation}
Moreover, $D(t,s)$ satisfies the following stochastic differential equation for given $ s $,
\begin{equation}\nonumber
	\frac{\rd D(t,s)}{D(t,s)}\!=\!(r_n(t)\!+\!\sigma_{c_1}\lambda_I\!+\!\sigma_{c_2}\lambda_S)\rd t\!+\!\sigma_{c_1}\rd W_I(t)\!+\!\sigma_{c_2}\rd W_S(t).
\end{equation}
Integrating $D(t,s)$ w.r.t. $s$ from time $t$ to time $T$, we construct a new process $F=\left\{F(t)|0\leqslant t\leqslant T\right\}$:
\[F(t)=\int_t^TD(t,s)\rd s,\quad\forall t\in [0,T].\]
$F(t)$ represents the expected value of the accumulated contribution rates from time $t$ to $T$ at time $t$. In particular, we have $$ F(T)=0,\ F(0)=c_0\int_{0}^T \exp[(\mu-\sigma_{c_1}\lambda_I-\sigma_{c_2}\lambda_S)s-\!\int_{0}^{s}r_n(u)\rd u]\rd s.$$
Based on the differential form of $D(t,s)$, the dynamics of $F$ are
\begin{equation}\nonumber
	\rd F(t)=-c(t)\rd t\!+\!F(t)[(r_n(t)\rd t\!+\!\sigma_{c_1}(\lambda_I\rd t\!+\!\rd W_I(t))\!+\!\sigma_{c_2}(\lambda_S\rd t\!+\!\rd W_S(t)))].
\end{equation}
Next, we construct the auxiliary process $\tilde{X}^{\tilde{\pi}}$ as follows:
\begin{equation}\label{defined}
	\tilde{X}^{\tilde{\pi}}(t) =\frac{X^{\pi}(t) +F(t)}{I(t)}.
\end{equation}
Because $F(T)=0$, $\tilde{X}^{\tilde{\pi}}(T)=\frac{X^{\pi}(T)}{I(T)}$. As such, the optimization problem over $\frac{X^{\pi}(T)}{I(T)}$ is equivalent to the optimization problem over $\tilde{X}^{\tilde{\pi}}(T)$. Besides, from Eqs.~(\ref{equ-X}) and (\ref{equ-I}), we have
\begin{equation}\label{SDE2}
	\begin{aligned}
		\rd \tilde{X}^{\tilde{\pi}}(t)=&\tilde{X}^{\tilde{\pi}}(t)[(\sigma_I^2+r_r(t))\rd t-\sigma_I(\lambda_I\rd t+\rd W_I(t))]\\
		&+\frac{1}{I(t)}[\pi_P(t)\sigma_I+\pi_S(t)\sigma_{S_1}+\sigma_{c_1}F(t)][\lambda_I\mathrm{d}t+\mathrm{d}W_I(t)]\\
		&+\frac{1}{I(t)}[\pi_S(t)\sigma_{S_2}+\sigma_{c_2}F(t)][\lambda_S\mathrm{d}t+\mathrm{d}W_S(t)]\\
		&-\frac{\sigma_I}{I(t)}[\sigma_{c_1}F(t)+\pi_P(t)\sigma_I+\pi_S(t)\sigma_{S_1}]\rd t,
	\end{aligned}
\end{equation}
and $ \tilde{X}^{\tilde{\pi}}(0)= \frac{X^{\pi}(0) +F(0)}{I(0)}$,  let $\tilde{x}_0=\tilde{X}^{\tilde{\pi}}(0)$.

For convenience, we make the following substitutions:
\begin{equation}\label{Transform}
	\left\{
	\begin{aligned}
		&\tilde{\pi}_S(t)\sigma_{S_2}=\frac{1}{I(t)}[\pi_S(t)\sigma_{S_2}+\sigma_{c_2}F(t)],\\
		&\tilde{\pi}_p(t)\sigma_{I}+\tilde{\pi}_S(t)\sigma_{S_1}=\frac{1}{I(t)}[\pi_P(t)\sigma_I+\pi_S(t)\sigma_{S_1}+\sigma_{c_1}F(t)-\sigma_{I}(X^{\pi}(t) +F(t))].
	\end{aligned}\right.
\end{equation}
Denote $\tilde{\pi}=(\tilde{\pi}_p,\tilde{\pi}_S)$ as the strategy associated with $\tilde{X}^{\tilde{\pi}}$. We call $\tilde{\pi}$ admissible if the related $\pi$ derived by Eq.~(\ref{Transform}) is in $\mathcal{A}(x_0)$. The admissible set of $\tilde{\pi}$ is denoted by $\mathcal{\tilde{A}}(\tilde{x}_0)$.

Then, Eq.~(\ref{SDE2})  is equivalent to
\begin{equation}\label{SDE3}
	\begin{aligned}
		\rd \tilde{X}^{\tilde{\pi}}(t)=&r_r(t)\tilde{X}^{\tilde{\pi}}(t)\rd t
		+[\tilde{\pi}_p(t)\sigma_{I}+\tilde{\pi}_S(t)\sigma_{S_1}][(\lambda_I-\sigma_{I})\mathrm{d}t+\mathrm{d}W_I(t)]\\
		&+\tilde{\pi}_S(t)\sigma_{S_2}[\lambda_S\mathrm{d}t+\mathrm{d}W_S(t)].
	\end{aligned}
\end{equation}
Observing Eq.~(\ref{SDE3}), we see that the state process of $\tilde{X}^{\tilde{\pi}}$ does not have additional cash flow, thus is self-financing. In fact, $\tilde{X}^{\tilde{\pi}}$ reflects the actual real available wealth of the fund. Besides, as Problem (\ref{PSV}) is concerned with real terminal wealth $\frac{X^{\pi}(T)}{I(T)}$ and $\frac{X^{\pi}(T)}{I(T)}=\tilde{X}^{\tilde{\pi}}(T)$, the optimization problem over  $\frac{X^{\pi}(T)}{I(T)}$ is equivalent to the problem over $\tilde{X}^{\tilde{\pi}}(T)$.  As such, Problem (\ref{PSV}) is equivalent to the following self-financing problem with only one state variable

 \begin{equation}\label{PSVD}
	\left\{\begin{aligned}
		\max_{\pi \in \tilde{\mathcal{A}}\left(\tilde{x}_{0}\right)}\quad&\frac{\mathbb{E}\left[U\left(\left(\tilde{X}^{\tilde{\pi}}(T)-{\theta}\right)_+\right)\right]}{\mathbb{E}\left[D\left(\left({\theta}-\tilde{X}^{\tilde{\pi}}(T)\right)_+\right)\right]},\\
		\text{s.t.}\quad&\tilde{X}^{\tilde{\pi}}(t)\text{ satisfies Eq.~(\ref{SDE3})},\\
		&\mathbb{P}(\tilde{X}^{\tilde{\pi}}(T)\geqslant L)\geqslant1-\varepsilon.
	\end{aligned}
	\right.\end{equation}
Using the auxiliary process, we transform the original problem (\ref{PSV}) into a single one w.r.t. $\tilde{X}^{\tilde{\pi}}$. On the one hand, the equivalent problem (\ref{PSVD}) is self-financing and martingale method can be well applied. On the other hand, $\tilde{X}^{\tilde{\pi}}$ contains the information of inflation and $I$ does not exist in Problem (\ref{PSVD}), which also largely decreases the computational complexity.
\vskip 10pt
\subsection{Martingale method}
In spite that Problem (\ref{PSVD}) is self-financing, the optimization goal of Problem (\ref{PSVD}) is non-linear with VaR constraint. As such, stochastic dynamic programming method can not be applied to solve Problem (\ref{PSVD}). Observing Eq.~(\ref{SDE3}), we see that there are two investment strategies and two Brownian motions. Therefore, $\tilde{X}^{\tilde{\pi}}$ can be viewed as the wealth process in a complete market with the pricing kernel process $H=\!\{H(t)|0\!\leqslant\! t\!\leqslant\! T\}$ given by

\begin{equation}\label{DeFofH}
	\frac{\rd H(t)}{H(t)}=-r_r(t)\rd t-(\lambda_I-\sigma_{I}) \rd W_I(t)-\lambda_S\rd W_S(t), H(0)=1.
\end{equation}
Following \cite{Cox}, we can transform the  optimization problem (\ref{PSVD})  into an equivalent one w.r.t. the terminal wealth $Z\triangleq \tilde{X}^{\tilde{\pi}}(T)$.

\begin{theorem}\label{Theorem1}
	The optimization problem (\ref{PSVD}) and the following problem (\ref{PVD}):
	\begin{equation}\label{PVD}
		\left\{\begin{aligned}
			\max_{Z \in \mathcal{M}_+}\quad&\frac{\mathbb{E}\left[U\left((Z-{\theta})_+\right)\right]}{\mathbb{E}\left[D\left(({\theta}-Z)_+\right)\right]},\\
			\text{s.t.}\quad&\mathbb{E}[H(T)Z]\leqslant\tilde{x}_0,\\
			&\mathbb{P}(Z\geqslant L)\geqslant1-\varepsilon,
		\end{aligned}
		\right.
	\end{equation}
have the same optimal values, where $ \mathcal{M}_+ $ denotes the set of non-negative $ \mathcal{F}_T $-measurable random variables.
\end{theorem}
\begin{proof}
The proof is simple and we omit it here, see \cite{Cox} for details.
\end{proof}
Theorem \ref{Theorem1} shows that the state process for $\tilde{X}^{\pi^*}$ in Problem (\ref{PSVD}) is equivalent to Problem (\ref{PVD}) with budget constraint over terminal wealth. Problem (\ref{PVD}) can be viewed as a static optimization problem over the random variable $Z\in\mathcal{M}_+$. In fact, the optimal $Z^*$ for Problem (\ref{PVD}) is related with the optimal terminal wealth in Problem  (\ref{PSVD}). As Eq.~(\ref{SDE3}) shows that $\tilde{X}$ can be treated as the wealth process in a complete market, the strategy within time horizon $[0,T]$ can be replicated.
\begin{proposition}\label{Propos1}
	Assume that $ Z^* $ is the optimal solution to Problem (\ref{PVD}). Then there exists a strategy $ \tilde{\pi}^*\in\tilde{\mathcal{A}}\left(\tilde{x}_{0}\right) $ such that  $ \tilde{X}^{\tilde{\pi}^*}(T)=Z^* $.
\end{proposition}
\begin{proof}
	As $ Z^* $ is the optimal solution to the Problem (\ref{PVD}), it is easy to see that the budget constraint $ \mathbb{E}[H(T)Z]\leqslant\tilde{x}_0$ should hold with equality, i.e., $ \mathbb{E}[H(T)Z^*]=\tilde{x}_0 $. Otherwise, consider \begin{equation*}
		\tilde{Z}\triangleq Z^*+e^{\int_{0}^{T}r_r(s)\mathrm{d}s}(\tilde{x}_0-\mathbb{E}[H(T)Z^*]).
	\end{equation*}
By investing the excess wealth in the risk-free asset, it is easy to verify that  $ \tilde{Z} $ is still a feasible solution to Problem (\ref{PVD}), while with a larger target value. This contradicts with the assumption that $ Z^* $ is the optimal solution to Problem (\ref{PVD}). Thus $ \mathbb{E}[H(T)Z^*]=\tilde{x}_0 $.
%

	Define the process \begin{equation}\label{Process1}
		Y^*(t)\triangleq  H(t)^{-1}\mathbb{E}[H(T)Z^*|\mathcal{F}_t],\quad0\leqslant t\leqslant T,
	\end{equation}
which is a martingale w.r.t. the filtration $ \mathbb{F} $. Based on the martingale representation theorem, there exists an $ \mathbb{F} $-progressively measurable process $\psi=\{ \psi(t)=(\psi_1(t),\psi_2(t))^\intercal|0\leqslant t\leqslant T\} $ such that  $ \int_{0}^T\|\psi(t)\|^2\mathrm{d}t<\infty,a.s. $ and \begin{equation}\label{HY}
		H(t)Y^*(t)=\tilde{x}_0 +\int_{0}^t\psi_1(s)\rd W_I(s)+\int_{0}^t\psi_2(s)\rd W_S(s),\quad0\leqslant t\leqslant T,~a.s..
	\end{equation}
Comparing  Eq.~(\ref{HY}) with Eqs.~(\ref{SDE3}) and (\ref{DeFofH}), let
\begin{equation*}
	\left\{\begin{aligned}
		\tilde\pi^*_S(t)=&\sigma_{S_2}^{-1}\left[H^{-1}(t)\psi_2(t)+\lambda_SY^*(t)\right],\\
		\tilde\pi^*_P(t)=&\sigma_{I}^{-1}\left[H^{-1}(t)\psi_1(t)+(\lambda_I-\sigma_{I})Y^*(t)-\tilde\pi^*_S(t)\sigma_{S_1}\right],
	\end{aligned}\right.
\end{equation*}
{	that is, \begin{equation}\label{Process2}
		\left\{\begin{aligned}
			\pi^*_S(t)=&\sigma_{S_2}^{-1}\left[I(t)H^{-1}(t)\psi_2(t)+\lambda_SI(t)Y^*(t)-\sigma_{C_2}F(t)\right],\\
			\pi^*_P(t)=&\sigma_{I}^{-1}\left[I(t)H^{-1}(t)\psi_1(t)+\lambda_II(t)Y^*(t)-\sigma_{C_1}F(t)-\pi^*_S(t)\sigma_{S_1}\right].
		\end{aligned}\right.
	\end{equation}
	Let $ \tilde\pi_S^*=\left\{\tilde\pi^*_S(t)|0\leqslant t\leqslant T\right\} $, $ \tilde\pi_P^*=\left\{\tilde\pi^*_P(t)|0\leqslant t\leqslant T\right\} $ and $ {\tilde\pi}^*= (\tilde\pi_P^*,\tilde\pi_S^*)$, then we have $ \tilde\pi^*\in\tilde{\mathcal{A}}\left(\tilde{x}_{0}\right) $ and $ \tilde{X}^{\pi^*}(t)=Y^*(t) ,~a.s.~\forall~0\leqslant t\leqslant T$. As such, we have $ \tilde{X}^{\tilde\pi^*}(T)=Y^*(T)=Z^* $.}
\end{proof}
Based on  Theorem \ref{Theorem1} and Proposition \ref{Propos1}, we see that the optimal values of Problem (\ref{PSVD}) and Problem (\ref{PVD}) are equal. In addition, the proof of Proposition \ref{Propos1} shows the relationship between the optimal solutions of these two problems. Therefore, instead of the stochastic dynamic optimization problem (\ref{PSVD}), we only need to investigate Problem (\ref{PVD}) first.
\vskip 10pt
\subsection{Linearization}
As the objective function in Problem (\ref{PVD}) is non-linear and different from the traditional maximization of  expected utility problem, and there are two constraints in the problem, we need to make further transformation to the problem. In this section, we transform the non-linear optimization problem into a linear optimization problem based on fractional programming. In order to ensure the well-posedness of the problem, we present some assumptions  about the reward function $U$ and the penalty function $ D $:

\begin{enumerate}
	\item[\bf (H1)]  $U$ and $D$ are both twice differentiable and strictly increasing functions with $U(0)=D(0)=0$.
\end{enumerate}
As the manager prefers large wealth, the monotonicity of $ U $ and $ D $ is a natural requirement. The condition that $U(0)=D(0)=0$ makes the performance ratio close to the Omega ratio considered by the predecessors when the terminal wealth is approximately equal to the reference level.
%
\begin{remark}\label{Nonconvex}
	According to \cite{JinZhou}, a sequence of random variables $ \{Z_n\}_{n\geqslant 1} $ can be constructed satisfying $ \mathbb{E}[H(T)Z_n]=\tilde{x}_0 $ with $ \mathbb{E}[Z_n]\to\infty $. As such, it is easy to see that if the reward function $ U $ is convex, the optimization problem (\ref{PVD}) is unbounded if without VaR constraints. Therefore, in order to ensure the well-posedness of the problem, we assume that the reward function $ U $ is concave, which is also a commonly used assumption in the expected utility maximization problem.	
%
\end{remark}
More specifically, we suppose that the reward function $ U $ satisfies the following conditions:
\begin{enumerate}
	\item [\bf (H2)]The reward function $ U $ satisfies the Inada condition:\begin{equation*}
		\lim_{x\searrow0}U'(x)=\infty,\quad\lim_{x\to\infty}U'(x)=0.
	\end{equation*}
	\item [\bf (H3)]The reward function $U$ is strictly concave, i.e.,  $ U''(z)<0,\forall z\in(0,\infty). $
\end{enumerate}
However, the penalty function $D$ is not required to be concave or convex. The property of $D$ largely affects the solution of Problem (\ref{PVD}). In the following, for convenience, we assume that $D$ is concave or strictly convex corresponding to the risk aversion case and risk seeking case, respectively.

The objective function in Problem (\ref{PVD}) is the ratio of two expectations and is non-linear. We adopt the linearization method by fractional programming to transform the problem into a family of linear optimization problems: for any parameter $\nu\geqslant0$, consider a family of optimization problems:
\begin{equation}\label{Linear}
	v(\nu,\tilde{x}_{0})\triangleq \sup_{Z\in\mathcal{C}(\tilde{x}_{0}),\mathbb{P}(Z\geqslant L)\geqslant1-\varepsilon}{\mathbb{E}\left[U\left((Z-{\theta})_+\right)\right]}-\nu{\mathbb{E}\left[D\left(({\theta}-Z)_+\right)\right]},
\end{equation}where\begin{equation*}
	\mathcal{C}(\tilde{x}_{0})=\left\{Z\in\mathcal{M}_+|~\mathbb{E}[H(T)Z]\leqslant\tilde{x}_0\right\}.
\end{equation*}
\begin{remark}\label{REMARKV}
	Similar to the proof of Proposition \ref{Propos1}, it is easy to verify that when Problem (\ref{Linear}) attains the optimal value, the budget constraint holds with equality, i.e.,	
	\begin{equation}\nonumber
		v(\nu,\tilde{x}_{0})=\sup_{Z\in\mathcal{M}_+,\mathbb{E}[H(T)Z]=\tilde{x}_0,\mathbb{P}(Z\geqslant L)\geqslant1-\varepsilon}{\mathbb{E}\left[U\left((Z-{\theta})_+\right)\right]}-\nu{\mathbb{E}\left[D\left(({\theta}-Z)_+\right)\right]}.
	\end{equation}
\end{remark}
Problem (\ref{Linear}) is linear and contains reward function and penalty function. In order to solve Problem  (\ref{PVD}), it is only necessary to solve a family of linearized problems (\ref{Linear}). Under suitable condition of $\nu$, the optimal solutions to Problems  (\ref{PVD}) and (\ref{Linear}) are the same, which is shown by the following theorem.
\begin{theorem}\label{Theorem2}
	Assume $ \tilde{x}_0<e^{-\int_{0}^{T}r_r(s)\rd s}\theta $, and that, for $ \nu\geqslant0 $, $ Z_{\nu} $ is the optimal solution of Problem (\ref{Linear}), i.e., \begin{equation}\nonumber
		v(\nu,\tilde{x}_{0})={\mathbb{E}\left[U\left((Z_\nu-{\theta})_+\right)\right]}-\nu{\mathbb{E}\left[D\left(({\theta}-Z_\nu)_+\right)\right]}.
	\end{equation}
	If there exists $ \nu^*\geqslant0 $ such that \begin{equation*}
		\nu^*=\frac{\mathbb{E}\left[U\left((Z_{\nu^*}-{\theta})_+\right)\right]}{\mathbb{E}\left[D\left(({\theta}-Z_{\nu^*})_+\right)\right]},
	\end{equation*}
	i.e., there exists $ \nu^*\geqslant0 $ such that $ v(\nu^*,\tilde{x}_{0})=0 $, then we have that $ Z^*\triangleq Z_{\nu^*} $ is the optimal solution to Problem (\ref{PVD}). Besides, $ \nu^* $ is the optimal valve of Problem (\ref{PVD}).
\end{theorem}

\begin{proof}
	Based on $ \tilde{x}_0<e^{-\int_{0}^{T}r_r(s)\rd s}\theta $ and the definition of $ \mathcal{C}(\tilde{x}_{0}) $, we know that $ \mathbb{P}(Z<\theta) >0$, otherwise, we have $ \tilde{x}_0\geqslant\mathbb{E}[H(T)Z] \geqslant\mathbb{E}[H(T)]\theta= e^{-\int_{0}^{T}r_r(s)\rd s}\theta$, which contradicts with $ \tilde{x}_0<e^{-\int_{0}^{T}r_r(s)\rd s}\theta $.
	
	Using the optimality of $ Z_{\nu^*}$, $ \forall Z\in\mathcal{M}_+\text{~satisfying~} \mathbb{E}[H(T)Z]\leqslant\tilde{x}_{0}\text{~and~} \mathbb{P}(Z\geqslant L)\geqslant1-\varepsilon$, we have	\begin{equation}\nonumber
		\begin{aligned}
			0=&v(\nu^*,\tilde{x}_{0})\\
			=&{\mathbb{E}\left[U\left((Z_{\nu^*}-{\theta})_+\right)\right]}-\nu^*{\mathbb{E}\left[D\left(({\theta}-Z_{\nu^*})_+\right)\right]}\\
			\geqslant&{\mathbb{E}\left[U\left((Z-{\theta})_+\right)\right]}-\nu^*{\mathbb{E}\left[D\left(({\theta}-Z)_+\right)\right]}\\
			=&{\mathbb{E}\left[U\left((Z-{\theta})_+\right)\right]}-\frac{\mathbb{E}\left[U\left((Z_{\nu^*}-{\theta})_+\right)\right]}{\mathbb{E}\left[D\left(({\theta}-Z_{\nu^*})_+\right)\right]}{\mathbb{E}\left[D\left(({\theta}-Z)_+\right)\right]},
		\end{aligned}
	\end{equation}
	which means that
	\begin{equation}\nonumber
		\frac{\mathbb{E}\left[U\left((Z_{\nu^*}-{\theta})_+\right)\right]}{\mathbb{E}\left[D\left(({\theta}-Z_{\nu^*})_+\right)\right]}\geqslant\frac{\mathbb{E}\left[U\left((Z-{\theta})_+\right)\right]}{\mathbb{E}\left[D\left(({\theta}-Z)_+\right)\right]}.
	\end{equation}
	As such, $ Z^*=Z_{\nu^*} $ is the optimal solution to Problem (\ref{PVD}) and $ \nu^* $ is the related optimal value. We call $\nu^* $ the optimal parameter.
\end{proof}
Based on Theorem \ref{Theorem2}, we need to disentangle Problem (\ref{PVD}) in the following two steps:
\begin{itemize}
	\item Prove the existence of the optimal parameter, i.e., the existence of the root of equation $ v(\cdot,\tilde{x}_{0})=0 $.
	\item Solve the optimization problem (\ref{Linear}) after linearization.
\end{itemize}
\vskip 10pt
\subsubsection{\bf Existence of optimal parameter $ \nu^* $}\label{EoOP}
In order to ensure the existence of $\nu^*$, we propose the following condition \textbf{(H4)}:
\begin{enumerate}
	\item [\bf (H4)]The reward function $ U $ has the asymptotic behavior of the ``Arrow-Pratt relative risk aversion" and satisfies the asymptotic elasticity condition:
	\begin{equation*}
		\liminf_{x\to\infty}\left(-\frac{xU''(x)}{U'(x)}\right)>0,\quad\lim_{x\to\infty}\frac{xU''(x)}{U'(x)}<1.
	\end{equation*}
\end{enumerate}
In the following, we will explore the properties of the function  $ v(\cdot,\tilde{x}_{0}) $ to  show the existence of the optimal parameter, i.e., the existence of the root of equation $ v(\cdot,\tilde{x}_{0})=0 $. We need the following lemma first.
\begin{lemma}\label{LEMMA1}
	Assume that {(H1)}-{(H4)} hold and let
	 \begin{equation*}\begin{aligned}
			M\triangleq\sup_{Z\in\mathcal{M}_+,\mathbb{E}[H(T)Z]=\tilde{x}_0}{\mathbb{E}\left[U\left((Z-{\theta})_+\right)\right]},\\m\triangleq\inf_{Z\in\mathcal{M}_+,\mathbb{E}[H(T)Z]=\tilde{x}_0}{\mathbb{E}\left[D\left(({\theta}-Z)_+\right)\right]}.
		\end{aligned}
	\end{equation*}Then if $ \tilde{x}_{0}<e^{-\int_{0}^{T}r_r(s)\rd s}\theta $, we have $ M<\infty$ and $m>0$.
\end{lemma}
\begin{proof}
	For the proof of $ M<\infty $, see \cite{JinZhouConv}.
	
Obviously, $m\geqslant 0$. If $m=0 $, then there exists $ Z_n\in\mathcal{M}_+$ with $ \mathbb{E}[H(T)Z_n]=\tilde{x}_0 $ such that $\displaystyle \lim_{n\to\infty}{\mathbb{E}\left[D\left(({\theta}-Z_n)_+\right)\right]}=0 $. Then $ D\left(({\theta}-Z_n)_+\right) $ converges to 0 in probability, thus $ ({\theta}-Z_n)_+ $ converges to 0 in probability. As such, $ ({\theta}-Z_n)_+ $ has a subsequence $ ({\theta}-Z_{n_m})_+ $ which converges to 0 $ a.s. $. Therefore, $ ({\theta}-Z_{n_m})_+ $ converges to 0 $ a.s. $ with respect to the new measure $ \mathbb{Q} $ which is given by \begin{equation*}
		\frac{\mathrm{d}\mathbb{Q}}{\mathrm{d}\mathbb{P}}=e^{\int_{0}^{T}r_r(s)\rd s} H(T).
	\end{equation*}
	As $ 0\leqslant({\theta}-Z_{n_m})_+\leqslant \theta$, $ ({\theta}-Z_{n_m})_+ $ converges to 0 in $ L^1 $-norm with respect to $ \mathbb{Q} $ based on {Dominated convergence theorem}, i.e.,  $\mathbb{E}^{\mathbb{Q}}\{({\theta}-Z_{n_m})_+\}\to0 $.
	Then  $ H(T)({\theta}-Z_{n_m} )_+$ converges to 0 in $L^1$-norm with respect to $ \mathbb{P} $, i.e.,  $ \mathbb{E}\{H(T)({\theta}-Z_{n_m} )_+\}\to0 $, which contradicts with the condition $ \mathbb{E}\{H(T)({\theta}-Z_n)\}= e^{-\int_{0}^{T}r_r(s)\rd s} \theta-\tilde{x}_{0}>0$.
	Thus $m>0 $.
\end{proof}
\begin{theorem}\label{nu}
	
	Suppose that {(H1)}-{(H4)} hold and the initial value $ \tilde{x}_{0} $ satisfies  \begin{equation}\label{condition}
		L\int_{q(\varepsilon)}^{+\infty}e^{-a-x}f(x)\rd x<\tilde{x}_{0}<e^{-\int_{0}^{T}r_r(s)\rd s}\theta,
	\end{equation}
	where $ a=\int_{0}^Tr_r(s)\rd s+\frac12((\lambda_I-\sigma_I)^2+\lambda_S^2)T $, $ q(\varepsilon) $ is the $ \varepsilon $-quantile of the normal  distribution $ N(0,((\lambda_I-\sigma_I)^2+\lambda_S^2)T) $, and $ f(\cdot) $ is the p.d.f. of the normal distribution $ N(0,((\lambda_I-\sigma_I)^2+\lambda_S^2)T) $. Then the function $ v(\cdot,\tilde{x}_{0}) $ admits the following properties
	
	\begin{enumerate}
		\item[(i)] $0< v(0,\tilde{x}_{0})<\infty$.
		\item[(ii)]$ v(\cdot,\tilde{x}_{0}) $ is non-increasing in $ (-\infty,+\infty) $.
		\item[(iii)] For a given $ \tilde{x}_{0}>0 $, $ v(\cdot,\tilde{x}_{0}) $ is convex in $(-\infty,+\infty) $.
		\item[(iv)] $\displaystyle\lim_{\nu\to\infty}v(\nu,\tilde{x}_{0})=-\infty.  $
	\end{enumerate}
\end{theorem}
To prove Theorem \ref{nu}, we first prove the following:
\begin{lemma}\label{Lx0}
	If Eq.~(\ref{Lx01}) \begin{equation}\label{Lx01}
		L\int_{q(\varepsilon)}^{+\infty}e^{-a-x}f(x)\rd x<\tilde{x}_{0}
	\end{equation}holds, where $ a ,q(\varepsilon)$ and $ f(\cdot) $ are the same as in Theorem \ref{nu}, then there exists $ Z\in\mathcal{M}_+ $ satisfying
	\begin{equation}\label{SOLVEable}\left\{\begin{aligned}
			\mathbb{E}[H(T)Z]\leqslant\tilde{x}_0,\\
			\mathbb{P}(Z\geqslant L)\geqslant1-\varepsilon,\\
			\mathbb{P}(Z> \theta)>0.
		\end{aligned}\right.
	\end{equation}
	Furthermore, if $ L<\theta $, then the condition (\ref{Lx01}) is equivalent to that there exists $ Z\in\mathcal{M}_+ $ satisfying Eq.~(\ref{SOLVEable}).
\end{lemma}
\begin{proof}
	Let \begin{equation*}
		X(\omega)=\left\{\begin{aligned}
			&L,&H(T)(\omega)\leqslant e^{-a-q(\varepsilon)},\\
			&0,&\text{otherwise}.
		\end{aligned}\right.
	\end{equation*}
	Then we have $ 	\mathbb{P}(X\geqslant L)\geqslant1-\varepsilon $. As $ H(T) $ satisfies Eq.~(\ref{DeFofH}), if $ L<\theta $, $\forall Z\in\mathcal{M}_+$ satisfying Eq.~(\ref{SOLVEable}), we have $ \tilde{x}_0\geqslant\mathbb{E}[H(T)Z]>\mathbb{E}[H(T)X] =L\int_{q(\varepsilon)}^{+\infty}e^{-a-x}f(x)\rd x$.
	
	Furthermore, if Eq.~(\ref{Lx01}) holds, we can modify $ X $ such that Eq.~(\ref{SOLVEable}) holds.
	
	Based on the above arguments, we see that if $ L\geqslant\theta $, condition $ 	L\int_{q(\varepsilon)}^{+\infty}e^{-a-x}f(x)\rd x\leqslant\tilde{x}_{0} $ is sufficient and necessary.
\end{proof}
\begin{proof}[\textbf{\emph{Proof of Theorem \ref{nu}}}]
	Based on the condition (\ref{condition}) and Lemma \ref{Lx0}, we know that there exists $ Z\in\mathcal{M}_+ $ satisfying Eq.~(\ref{SOLVEable}), which means that
	 Problem (\ref{PVD}) has a non-trivial feasible solution. As such, we get $v(0,\tilde{x}_{0})>0$. In addition, $ v(\nu,\tilde{x}_{0})\leqslant M-\nu m$ can be derived by Lemma \ref{LEMMA1} and Remark \ref{REMARKV}. Then we have $ v(0,\tilde{x}_{0})<\infty$ and $\displaystyle\lim_{\nu\to\infty}v(\nu,\tilde{x}_{0})=-\infty$.
	
	$\forall \nu_1<\nu_2 ,t\in[0,1]$, assume that $ Z_{\nu_2} $ maximizes $ v(\nu_2,\tilde{x}_{0})$ and $ Z_{t} $ maximizes $ v(t\nu_1+(1-t)\nu_2,\tilde{x}_{0})$. Then we have \begin{equation}
		\begin{aligned}\nonumber
			v(\nu_2,\tilde{x}_{0})&= \mathbb{E}\left[U\left((Z_{\nu_2}-{\theta})_+\right)-\nu_2D\left(({\theta}-Z_{\nu_2})_+\right)\right]\\
			&\leqslant \mathbb{E}\left[U\left((Z_{\nu_2}-{\theta})_+\right)-\nu_1D\left(({\theta}-Z_{\nu_2})_+\right)\right]\\
			&\leqslant v(\nu_1,\tilde{x}_{0}),
		\end{aligned}
	\end{equation}
	i.e., $ v(\cdot,\tilde{x}_{0}) $ is a non-increasing function.
	Furthermore, \begin{equation}\nonumber
		\begin{aligned}
			v(t\nu_1+(1-t)\nu_2,\tilde{x}_{0})=&  \mathbb{E}\left[U\left((Z_{t}-{\theta})_+\right)-(t\nu_1+(1-t)\nu_2)D\left(({\theta}-Z_{t})_+\right)\right]\\
			=&t \mathbb{E}\left[U\left((Z_{t}-{\theta})_+\right)-\nu_1D\left(({\theta}-Z_{t})_+\right)\right]\\
			&+(1-t)\mathbb{E}\left[U\left((Z_{t}-{\theta})_+\right)-\nu_2D\left(({\theta}-Z_{t})_+\right)\right]\\
			\leqslant& tv(\nu_1,\tilde{x}_{0})+(1-t)v(\nu_2,\tilde{x}_{0}).
		\end{aligned}
	\end{equation}
	As such, $ v(\cdot,\tilde{x}_{0}) $ is convex in $ \mathbb{R} $.
\end{proof}
\begin{corollary}\label{Corollary1}
	There exists a unique $ \nu^*>0 $ such that $ v(\nu^*,\tilde{x}_{0})=0 $.
\end{corollary}
\begin{proof}
	Because $ v(\cdot,\tilde{x}_{0}) $ is convex in $ \mathbb{R} $, it is continuous (as $ \mathbb{R}$ is a convex open set). As $0< v(0,\tilde{x}_{0})<\infty$ and  $\displaystyle\lim_{\nu\to\infty}v(\nu,\tilde{x}_{0})=-\infty$, using the
	intermediate value theorem, we know that there exists $ \nu^* $ such that $ v(\nu^*,\tilde{x}_{0})=0 $. Suppose that $ \nu_1<\nu_2 $ are both the roots of $ v(\cdot,\tilde{x}_{0}) =0$. As $\displaystyle\lim_{\nu\to\infty}v(\nu,\tilde{x}_{0})=-\infty $, we can find some $ \nu_3>\nu_2 $ such that $ v(\nu_3,\tilde{x}_{0}) < v(\nu_2,\tilde{x}_{0})=0 $. As such, noticing \begin{equation}\nonumber
		\nu_2=\frac{\nu_2-\nu_1}{\nu_3-\nu_1}\nu_3+\frac{\nu_3-\nu_2}{\nu_3-\nu_1}\nu_1,
	\end{equation} 
we have \begin{equation}\nonumber
		v(\nu_2,\tilde{x}_{0})=0>\frac{\nu_2-\nu_1}{\nu_3-\nu_1} v(\nu_3,\tilde{x}_{0})+\frac{\nu_3-\nu_2}{\nu_3-\nu_1}v(\nu_1,\tilde{x}_{0}),
	\end{equation}which contradicts with the convexity of $ v(\cdot,\tilde{x}_{0}) $ in $ \mathbb{R} $. Thus, the root $ \nu^* $ is unique.
\end{proof}
\begin{remark}
	The constraint (\ref{condition}) is also a necessary condition to ensure that Problem (\ref{PVD}) has a non-trivial feasible solution.
\end{remark}
Next, we only need to solve the optimization problem (\ref{Linear}) after linearization. Problem (\ref{Linear}) searches the optimal random variable under two constraints. We transform the constrained problem into problem with no constraint based on Lagrange dual method.
\vskip 10pt
\subsubsection{\bf Lagrange dual method}
Problem (\ref{Linear}) can be viewed as the expected utility maximization problem under two constraints. Using Lagrange dual method, we first eliminate the VaR constraint, then eliminate the budget constraint. Afterwards, we solve the non-constrained problem based on concavification method.
\vskip 5pt
Define \begin{align*}
	f_{\nu}(Z)&\triangleq U\left((Z-{\theta})_+\right)-\nu D\left(({\theta}-Z)_+\right),\\
	f_{\nu,\lambda}(Z)&\triangleq U\left((Z-{\theta})_+\right)-\nu D\left(({\theta}-Z)_+\right)+\lambda\mathbf{1}_{Z\geqslant L},
\end{align*}
for the optimization problem (\ref{Linear}) after linearization:
\begin{equation}\label{LiVaR}
	\left\{\begin{aligned}
		\max_{Z \in \mathcal{M}_+}\quad& \mathbb{E}\{f_{\nu}(Z)\},\\
		\text{s.t.}\quad &\mathbb{E}[H(T)Z]\leqslant\tilde{x}_{0}\\
		&P(Z\geqslant L)\geqslant1-\varepsilon.
\end{aligned}\right.\end{equation}
Using Lagrange dual method, we obtain the following equivalent optimization problem without VaR constraint:
\begin{equation}\label{Li}
	\left\{\begin{aligned}
		\displaystyle\max_{Z \in \mathcal{M}_+}\quad& \mathbb{E}\{f_{\nu,\lambda}(Z)\},\\
		\text{   s.t.}\quad &\mathbb{E}[H(T)Z]\leqslant\tilde{x}_{0}.
\end{aligned}\right.\end{equation}
Problems (\ref{Li}) and (\ref{LiVaR}) are equivalent. The following theorem shows the relationship between the new optimization problem (\ref{Li}) and Problem (\ref{LiVaR}).
\begin{theorem}\label{VaR}
	For $ \lambda\geqslant0 $, let $ Z_{\nu,\lambda} $ be the solution to Problem (\ref{Li}). If there exists a Lagrange multiplier $ \lambda^*\geqslant0 $ such that \begin{align}
		&P(Z_{\nu,\lambda^*}\geqslant L)\geqslant1-\varepsilon\label{TranCond1},\\
		&\lambda^*\cdot[P(Z_{\nu,\lambda^*}\geqslant L)-(1-\varepsilon)]=0\label{TranCond2}.
	\end{align}Then $Z_{\nu}\triangleq Z_{\nu,\lambda^*} $ is the optimal solution to Problem (\ref{LiVaR}).
\end{theorem}
\begin{proof}
	First, using the condition (\ref{TranCond1}), we know that $ Z_{\nu,\lambda^*} $ is a feasible solution to Problem (\ref{LiVaR}). As such, $ \mathbb{E}\{f_{\nu}(Z_{\nu,\lambda^*})\} $ does not exceed the optimal value of the optimization problem (\ref{LiVaR}).
	Moreover, for any feasible solution $ Z $ of Problem (\ref{LiVaR}), $ Z $ is also a feasible solution of Problem (\ref{Li}), as such, $ \mathbb{E}\{f_{\nu,\lambda^*}(Z)\}\leqslant\mathbb{E}\{f_{\nu,\lambda^*}(Z_{\nu,\lambda^*})\} $, i.e., \begin{equation}\nonumber
		\mathbb{E}\{f_{\nu}(Z)\}+\lambda^*P(Z\geqslant L)\leqslant\mathbb{E}\{f_{\nu}(Z_{\nu,\lambda^*})\}+\lambda^*P(Z_{\nu,\lambda^*}\geqslant L).
	\end{equation}
	Then
	\begin{equation}\begin{aligned}\nonumber
			\mathbb{E}\{f_{\nu}(Z)\}\leqslant&\mathbb{E}\{f_{\nu}(Z_{\nu,\lambda^*})\}+\lambda^*\left[P(Z_{\nu,\lambda^*}\geqslant L)-P(Z\geqslant L)\right]\\
			=&\mathbb{E}\{f_{\nu}(Z_{\nu,\lambda^*})\}+\lambda^*\left[1-\varepsilon-P(Z\geqslant L)\right]\\
			\leqslant&\mathbb{E}\{f_{\nu}(Z_{\nu,\lambda^*})\}.
		\end{aligned}
	\end{equation}
	As such, the optimal value of Problem (\ref{LiVaR}) does not exceed $ \mathbb{E}\{f_{\nu}(Z_{\nu,\lambda^*})\} $. Thus, $ Z_{\nu,\lambda^*} $ is the optimal solution of Problem (\ref{LiVaR}).
	
\end{proof}
\vskip 5pt
\begin{remark}
	The conditions to ensure the existence of $ \lambda^*\geqslant0 $ satisfying Theorem \ref{VaR} will be discussed later.
\end{remark}

Problem (\ref{Li}) still has the budget constraint  $ \mathbb{E}[H(T)Z]\leqslant\tilde{x}_{0} $. Similar with the transformation of the VaR constraint, we transform Problem (\ref{Li}) into an equivalent one without budget constraint based on Lagrange dual method.
For each $ \beta>0 $, consider the following optimization problem: \begin{equation}\label{LiIn}
	\max_{Z \in \mathcal{M}_+}\ \mathbb{E}\{f_{\nu,\lambda}(Z)-\beta H(T)Z\}.\end{equation}
Problem (\ref{LiIn}) can be viewed as a static optimization problem. The related non-randomized version of Problem (\ref{LiIn}) is: for each $ y>0 $
\begin{equation}\label{L}
	\max_{x\in\mathbb{R}_+}\ \{f_{\nu,\lambda}(x)-yx\}.\end{equation}
We have the following theorem to depict the relationships among the optimal solutions of Problems (\ref{Li}), (\ref{LiIn}) and (\ref{L}):
\begin{theorem}\label{DeTr}
	For all $ \nu\geqslant0$ and $\lambda\geqslant0 $, we have the following properties:
	\begin{enumerate}
		\item [(a)]For all $ \nu\geqslant0,\lambda\geqslant0 $ and $ y>0 $, let the Borel measurable function $ x_{\nu,\lambda}^*(y) $ be the optimal solution to Problem (\ref{L}). Then $ Z_{\nu,\lambda,\beta}\triangleq x_{\nu,\lambda}^*(\beta H(T)) $ is the optimal solution to Problem (\ref{LiIn}).
		\item [(b)]If there exists $ \beta^*>0 $ such that $ Z_{\nu,\lambda,\beta^*} $ is the optimal solution of Problem (\ref{LiIn}) with $ \mathbb{E}[H(T)Z_{\nu,\lambda,\beta^*}]=\tilde{x}_{0}  $, then $ Z_{\nu,\lambda}\triangleq Z_{\nu,\lambda,\beta^*} $ is the optimal  solution to Problem (\ref{Li}). Moreover, the $ \beta^* $ is called the optimal multiplier.
	\end{enumerate}
\end{theorem}
\begin{proof}
	See \cite{SaundersLin} for the proof.
\end{proof}
\vskip 5pt
Next, the following two steps are essential:
\begin{enumerate}
	\item Prove the existence of the optimal multiplier $ \beta^* $.  The existence of $ \beta^*>0 $ ensures that $ Z_{\nu,\lambda,\beta^*}$ is the optimal solution of Problem (\ref{LiIn}).
%
%
	\item Disentangle the static optimization problem (\ref{L}).
\end{enumerate}

\vskip 10pt
\subsubsection{\bf Existence of the optimal multipliers $ \beta^* $ and $ \lambda^*$ in Theorems~\ref{VaR} and \ref{DeTr}}
We first prove the existence of $ \beta^* $ when the parameter $ \nu$ and $\lambda $ are given. Similar to the proof of existence of the optimal parameter $ \nu^* $ in Subsection \ref{EoOP}, based on Theorem \ref{DeTr}, define
\begin{equation}\nonumber
	R_{\nu}(\lambda,\beta)\triangleq\mathbb{E}[H(T)Z_{\nu,\lambda,\beta}]\equiv\mathbb{E}[H(T)x_{\nu,\lambda}^*(\beta H(T))].
\end{equation}
We will verify that the function $ R_{\nu}(\lambda,\beta) $ has the following properties.
\begin{proposition}\label{pro-R}
	Assume that {(H1)}-{(H4)} hold. We have the following properties:
	\begin{enumerate}\label{properties}
		\item[(a)] $0< R_{\nu}(\lambda,\beta)<\infty$.
		\item[(b)]$ R_{\nu}(\lambda,\cdot) $ is continuous in $ (-\infty,+\infty)$.
		\item[(c)] $\displaystyle \lim_{\beta\to\infty}R_{\nu}(\lambda,\beta)=0$, $\displaystyle\lim_{\beta\searrow0}R_{\nu}(\lambda,\beta)=\infty $.
	\end{enumerate}
\end{proposition}
\begin{proof}
	As the solution $ x_{\nu,\lambda}^*(y) $ is non-increasing in $y$, the following relation holds\begin{equation}\nonumber
		R_{\nu}(\lambda,\beta)=\mathbb{E}[H(T)x_{\nu,\lambda}^*(\beta H(T))]
		\leqslant\mathbb{E}\left\{H(T)\left[I_1(\beta H(T))+\theta\right]\right\}
		<\infty,
	\end{equation}where $ I_1(x) $ is the inverse function of $ U'(x) $ and will be discussed in detail in the following paragraph.
The proof of $ \mathbb{E}[H(T)I_1(\beta H(T))] <\infty$ relies on Assumption {(H4)}: the reward function $ U $ has the asymptotic behavior of the ``Arrow-Pratt relative risk aversion" and satisfies the asymptotic elasticity condition, see Corollary 5.1 in \cite{JinZhouConv} for detailed  proof.

Because $ x_{\nu,\lambda}^*(y) $ is non-negative and satisfies $\displaystyle\lim_{y\searrow0} x_{\nu,\lambda}^*(y)=\infty $ and $\displaystyle\lim_{y\to\infty} x_{\nu,\lambda}^*(y)=0 $, thus, using the monotonic convergence theorem, we have $\displaystyle \lim_{\beta\to\infty}R_{\nu}(\lambda,\beta)=0$, $\displaystyle\lim_{\beta\searrow0}R_{\nu}(\lambda,\beta)=\infty $.
	
	The proof of the continuity of $ R_{\nu}(\lambda,\beta) $ w.r.t $ \beta $ is based on Dominated convergence theorem: $ \forall\beta>0 $ and $ 0<\beta_n\to\beta $, for any given $ \epsilon>0 $, there exists $ N $ large enough, such that for any $ n>N $, the following equation holds:\begin{equation}\nonumber
		0\leqslant H(T)x_{\nu,\lambda}^*(\beta H(T))\leqslant H(T)\left[I_1((\beta-\epsilon) H(T))+\theta\right].
	\end{equation}
	Moreover, because the upper bound $ H(T)\left[I_1((\beta-\epsilon) H(T))+\theta\right] $ is integrable, using Dominated  convergence theorem, we have \begin{equation*}
		\lim_{\beta_n\to\beta}R_{\nu}(\lambda,\beta_n)=\lim_{\beta_n\to\beta}\mathbb{E}[H(T)x_{\nu,\lambda}^*(\beta_n H(T))]=\mathbb{E}\left[\lim_{\beta_n\to\beta}H(T)x_{\nu,\lambda}^*(\beta_n H(T))\right].\end{equation*}
	As the solution $ x_{\nu,\lambda}^*(\cdot) $ is continuous almost everywhere, we have\begin{equation*}
		\mathbb{E}\left[\lim_{\beta_n\to\beta}H(T)x_{\nu,\lambda}^*(\beta_n H(T))\right]=
		\mathbb{E}[H(T)x_{\nu,\lambda}^*(\beta H(T))]=R_{\nu}(\lambda,\beta).
	\end{equation*}
	As such, we obtain $\displaystyle \lim_{\beta_n\to\beta}R_{\nu}(\lambda,\beta_n)=R_{\nu}(\lambda,\beta) $, thus the continuity of $R_{\nu}(\lambda,\cdot)$ follows.
%
\end{proof}

Based on the monotonicity of $ x_{\nu,\lambda}^*(y) $ with respect to the variable $ y $, we can get the monotonicity of $ R_{\nu}(\lambda,\beta) $ with respect to $ \beta $, which plays an important role in the numerical calculation of $\beta^*$
%

Using the properties of $ R_{\nu}(\lambda,\beta)$ in Proposition \ref{pro-R} and the intermediate value theorem, we have the following corollary: \begin{corollary}\label{Corollary2}
	Assume that {(H1)}-{(H4)} hold,
	given $ \nu$ and $\lambda $, there exists a unique optimal multiplier $ \beta^* $ such that $ R_{\nu}(\lambda,\beta^*)=\mathbb{E}[H(T)Z_{\nu,\lambda,\beta^*}]=\tilde{x}_{0} $.
\end{corollary}

Next, we discuss the existence of $\lambda^* $  when the parameters $ \nu$ and  $ \beta $ are given.
Similar with \cite{OIWS}, the existence of $\lambda^* $ such that $ \lambda^*(\mathbb{P}[x_{\nu,\lambda^*}^*(\beta H(T))\geqslant L]-(1-\varepsilon)) =0$ can be proved if $ \beta $ is fixed. However,  Corollary \ref{Corollary2} shows that the $ \beta^* $ satisfying $ R_{\nu}(\lambda,\beta^*)=\tilde{x}_{0} $ depends on $ \lambda $. The existences of $\lambda^* $ and $ \beta^* $ may not hold at the same time in some cases.


Based on Corollary \ref{Corollary2}, the multiplier $ \beta^* $ can be regard as a function of $ \nu,\lambda,\tilde{x}_{0} $, as such, we define $ \beta^*\triangleq B_{\nu}(\lambda,\tilde{x}_{0}) $.
Similarly, based on Theorem \ref{DeTr}, we also define
\begin{equation}\nonumber
	S_{\nu}(\lambda,\beta)\triangleq\mathbb{P}[Z_{\nu,\lambda,\beta}\geqslant L]\equiv\mathbb{P}[x_{\nu,\lambda}^*(\beta H(T))\geqslant L].
\end{equation}
Similar to the proof of the continuity of $ R_{\nu}(\lambda,\beta) $ with respect to $ \beta $, we can get the continuity of $ B_{\nu}(\lambda,\tilde{x}_{0}) $ with respect to $ \lambda $ and  the continuity of $ S_{\nu}(\lambda,\beta) $ with respect to $ \lambda $ and $ \beta $. Thus, the continuity of $ S_{\nu}(\lambda,B_{\nu}(\lambda,\tilde{x}_{0})) $ with respect to $ \lambda $ also follows.
\vskip 5pt
For the existence of the optimal multipliers $ \beta^* $ and $ \lambda^*$, we list the different cases in the following proposition.
\begin{proposition}\label{prop-S}
	The existence of the optimal multipliers $ \beta^* $ and $ \lambda^*$ depends on the function $S_{\nu}(\lambda,B_{\nu}(\lambda,\tilde{x}_{0}))$.
	\begin{enumerate}
		\item If $ S_{\nu}(0,B_{\nu}(0,\tilde{x}_{0})) \geqslant1-\varepsilon$, then $ \beta^*=B_{\nu}(0,\tilde{x}_{0}) $ and $\lambda^*=0  $.
		\item If $ S_{\nu}(0,B_{\nu}(0,\tilde{x}_{0})) <1-\varepsilon$, define  \begin{equation*}
			p=\sup_{\lambda\geqslant0}S_{\nu}(\lambda,B_{\nu}(\lambda,\tilde{x}_{0})).
		\end{equation*}
	    \begin{enumerate}
	    	\item If $ p>1-\varepsilon $, then there exists $ \lambda^*>0 $ such that $ S_{\nu}(\lambda^*,B_{\nu}(\lambda^*,\tilde{x}_{0}))=1-\varepsilon $. And the multipliers  $\lambda^*  $ and $ \beta^*=B_{\nu}(\lambda^*,\tilde{x}_{0}) $ are what we need.
	    	\item If  $ p<1-\varepsilon $ or $ 1-\varepsilon  $ can not be achieved by any $ \lambda>0 $, then the optimal  multipliers $ \beta^* $ and $ \lambda^*$ do not exist.
	    \end{enumerate}
	\end{enumerate}
\end{proposition}
Based on Proposition \ref{prop-S}, we see that  the  budget and VaR constraints may not be satisfied at the same time, as such, there may be no optimal solution to the primal problem.
%

%
\vskip 10pt
\subsection{Optimal solution of Problem (\ref{L})}\label{ss:o}
Problem (\ref{L}) is a static optimization problem. However, $f_{\nu,\lambda}(\cdot)$ is not always globally concave while depends on the relationship between the functions $U$ and $D$. As such, Problem (\ref{L}) can not be directly solved by the first-order condition. We employ the concavification method as in \cite{He2018PROFIT} to solve the non-concave optimization problem (\ref{L}). The concave envelop $f^c(\cdot)$ of a give function $f(x), x\in G$, which is defined by
\[f^c(\cdot)\triangleq \inf\{g(\cdot)|g(\cdot):\ G\rightarrow \mathbb{R} \text{ is concave and } g(x)\ge f(x),\ \forall x\in G\}.\]
Then, based on \cite{He2018PROFIT} and \cite{POWPR}, Problem (\ref{L}) is equivalent to the following concavified problem
\begin{equation}\label{Lc}
	\max_{x\in\mathbb{R}_+}\quad \left\{f^c_{\nu,\lambda}(x)-yx\right\}.\end{equation}
In order to solve Problem (\ref{L}), we only need to derive the concave envelop of $f_{\nu,\lambda}(\cdot)$ first. Then we solve Problem  (\ref{Lc}) as the optimal solutions of Problems  (\ref{L}) and  (\ref{Lc})  are the same.

\vskip 5pt
The concave envelop of $f_{\nu,\lambda}(\cdot)$ relies on the relationship between $U$ and $D$. The property of $D$ also affects the concave envelop. We present the results for the optimal solution of the non-random optimization problem (\ref{L}) for convex and concave $D$ separately. When $D$ is convex, there are fourteen cases. When $D$ is concave, there are six cases.
\begin{theorem}\label{DET1}
	The solution of the non-random optimization problem (\ref{L}) is given in the following subsections \ref{ConvexPenaltyFunction} and \ref{ConcavePenaltyFunction}.
\end{theorem}
\vskip 5pt
\subsubsection{\bf Convex penalty function}\label{ConvexPenaltyFunction}
In the case of convex penalty function, for the linearized problem, we see that the utility function is piece-wise concave, which shows that the manager is always risk aversion while has different attitudes towards gains and losses.
%
Let $I_1(x)$ be the inverse function of function $U'(x)$, and $I_2(x)$ be the inverse function of function $D'(x)$. Denote \begin{equation}\label{knulambda}
	k_{\nu,\lambda} \triangleq \frac{f_{\nu}(L)+\lambda+\nu D(\theta)}{L}.
\end{equation}
As shown in Fig.~\ref{fig:xfig},  there exist unique $ z_1<\theta,z_2>\theta $ such that \begin{equation}\nonumber
	\nu D'(\theta-z_1)=U'(z_2-\theta)=\frac{U(z_2-\theta)+\nu D(\theta-z_1)}{(z_2-\theta)+(\theta-z_1)}.
\end{equation}
\begin{figure}[htbp] 
	\centering
	\includegraphics[width=0.7\linewidth]{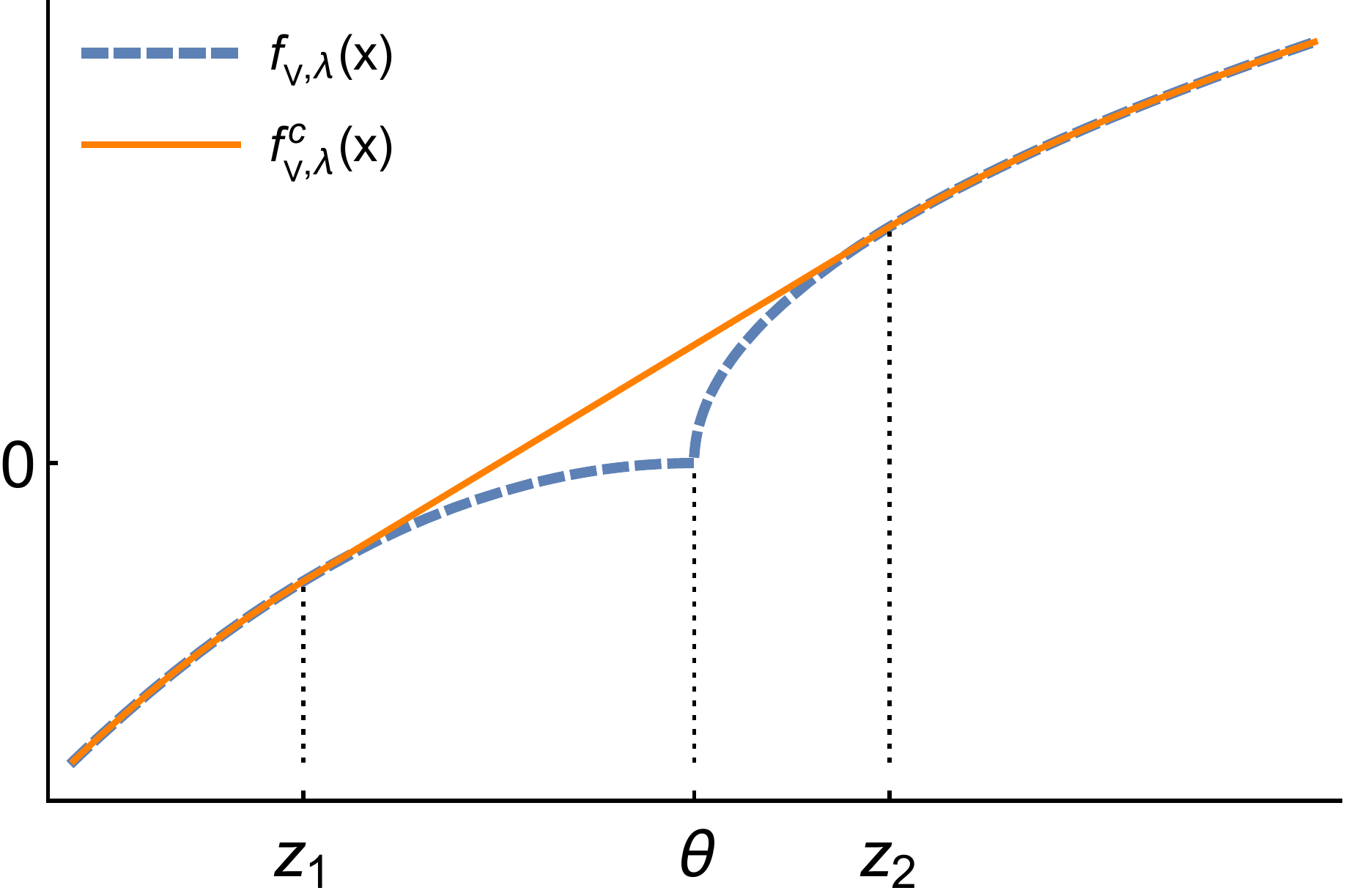}
	\caption{Schematic diagram of function $ f_{\nu}(x) $ and its concave envelope function.}
	\label{fig:xfig}
\end{figure}
\vskip 5pt
First we show the results when the VaR constraint is not binding, i.e., $\lambda=0$.  When  $ \lambda=0 $, the function $ f_{\nu}(x) $ is also equal to the function $ f_{\nu,\lambda}(x) $. If $ z_1\geqslant0 $, it is easy to find the solution to the non-random problem (\ref{L}) in this case  is\begin{equation}\nonumber x_{\nu,\lambda}^*(y)=\left\{
	\begin{aligned}
		& I_1(y)+\theta,&0<&y<U'(z_2-\theta),\\
		&\theta-I_2(\frac{y}{\nu}),&\nu D'(\theta-z_1)\leqslant &y<\nu D'(\theta),\\
		& 0,&\nu D'(\theta)\leqslant &y,
	\end{aligned}\right.
\end{equation}
which we call {\bf Case I} ($\lambda=0$ and $ z_1\geqslant0 $).

When $\lambda=0$ and $ z_1<0 $, there exists a unique $ z'>\theta $  satisfying \begin{equation}\nonumber
	U'(z'-\theta)=\frac{U(z'-\theta)+\nu D(\theta)}{z'-0}.
\end{equation}Then the optimal solution to the non-random problem (\ref{L}) in this case is
\begin{equation}\nonumber x_{\nu,\lambda}^*(y)=\left\{
	\begin{aligned}
		& I_1(y)+\theta,&0<&y<U'(z'-\theta),\\
		& 0,&U'(z'-\theta)\leqslant &y,
	\end{aligned}\right.
\end{equation}
which we call {\bf Case II }($\lambda=0$ and $ z_1<0 $).
\vskip 5pt
When $ \lambda>0 $, the optimal solution relies on the relationships among $L, \theta, z_1$ and $z_2$. We consider four specific cases: $L<z_1$, $z_1\leqslant L<\theta$, $\theta= L$, $\theta< L<z_2 $. We depict the concave envelope function of $f_{\nu,\lambda}(x) $ correspondingly and use it to solve the non-random problems (\ref{L}). We only show the results and corresponding schematic diagrams. The detailed calculations and proofs are similar with \cite{OIWS}. For simplicity, we omit the proofs here.

\paragraph{When $ L<z_1 $}We consider the relationship between $ k_{\nu,\lambda} $ and $ \nu D'(\theta) $ as follows:\begin{enumerate}
	\item If $ k_{\nu,\lambda}> \nu D'(\theta) $, then the schematic diagram is shown in Fig.~\ref{fig:xfig1}.
	\begin{figure}[htbp] 
		\centering
		\includegraphics[width=0.6\linewidth]{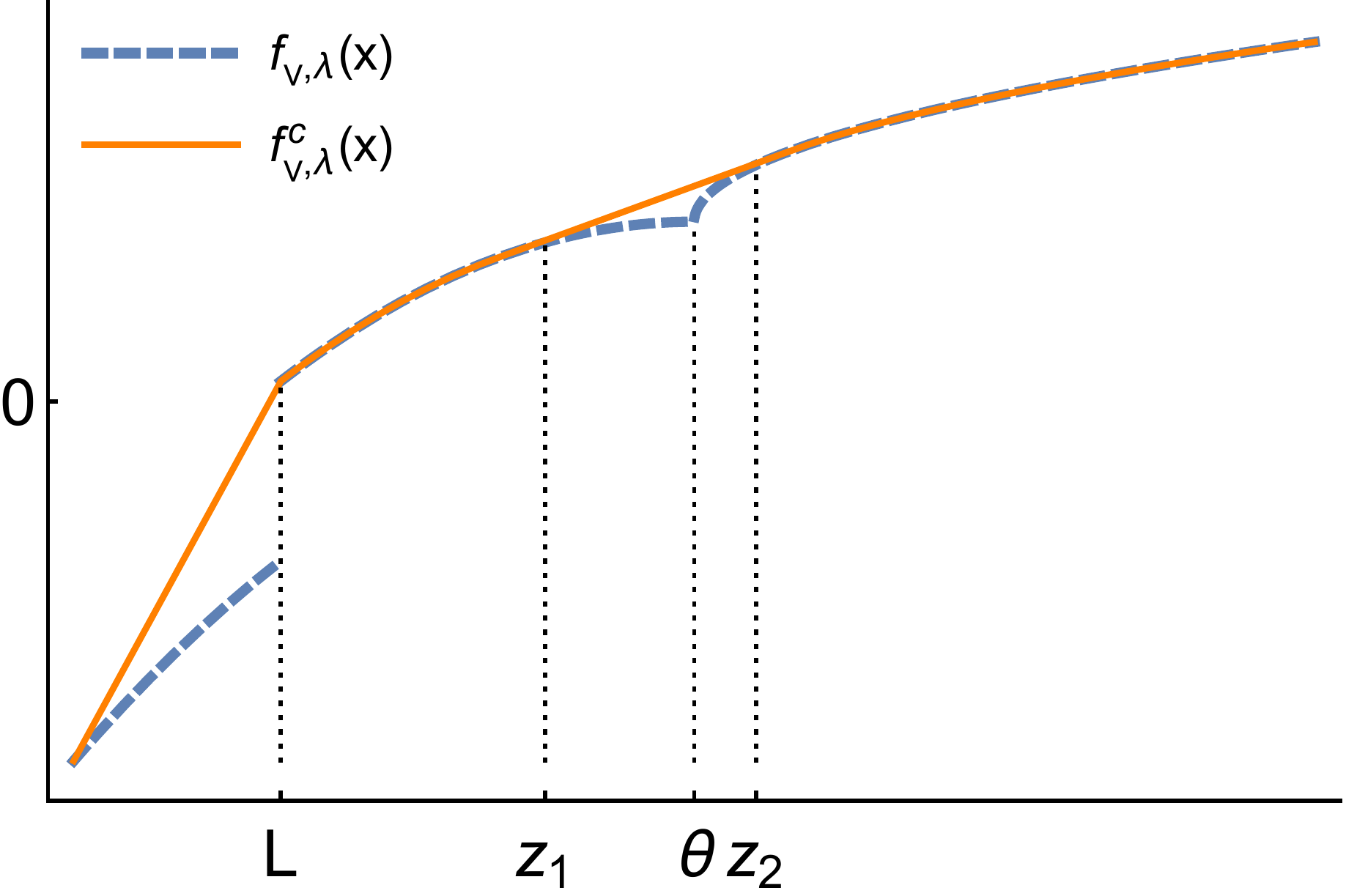}
		\caption{Schematic diagram of function $ f_{\nu,\lambda}(x) $ and its concave envelope function.}
		\label{fig:xfig1}
	\end{figure}

	And the optimal solution is
	\begin{equation}\nonumber x_{\nu,\lambda}^*(y)=\left\{
		\begin{aligned}
			& I_1(y)+\theta,&0<&y<U'(z_2-\theta),\\
			&\theta-I_2(\frac{y}{\nu}),&\nu D'(\theta-z_1)\leqslant &y<\nu D'(\theta-L),\\
			& L,&\nu D'(\theta-L)\leqslant &y< k_{\nu,\lambda},\\
			& 0,& k_{\nu,\lambda}\leqslant &y.
		\end{aligned}\right.
	\end{equation}
   We call this {\bf Case III }($ L<z_1 $ and $ k_{\nu,\lambda}> \nu D'(\theta) $).
	\item If $ k_{\nu,\lambda} \leqslant \nu D'(\theta) $, there exists a unique $  z_3\in[0,L) $ satisfying
	\begin{equation}\label{z3}
		\nu D'(\theta-z_3)=\frac{f_{\nu}(L)+\lambda+\nu D(\theta-z_3)}{L-z_3}.
	\end{equation}
	The schematic diagram is as shown in Fig.~\ref{fig:xfig2}.
	\begin{figure}[htbp] 
		\centering
		\includegraphics[width=0.6\linewidth]{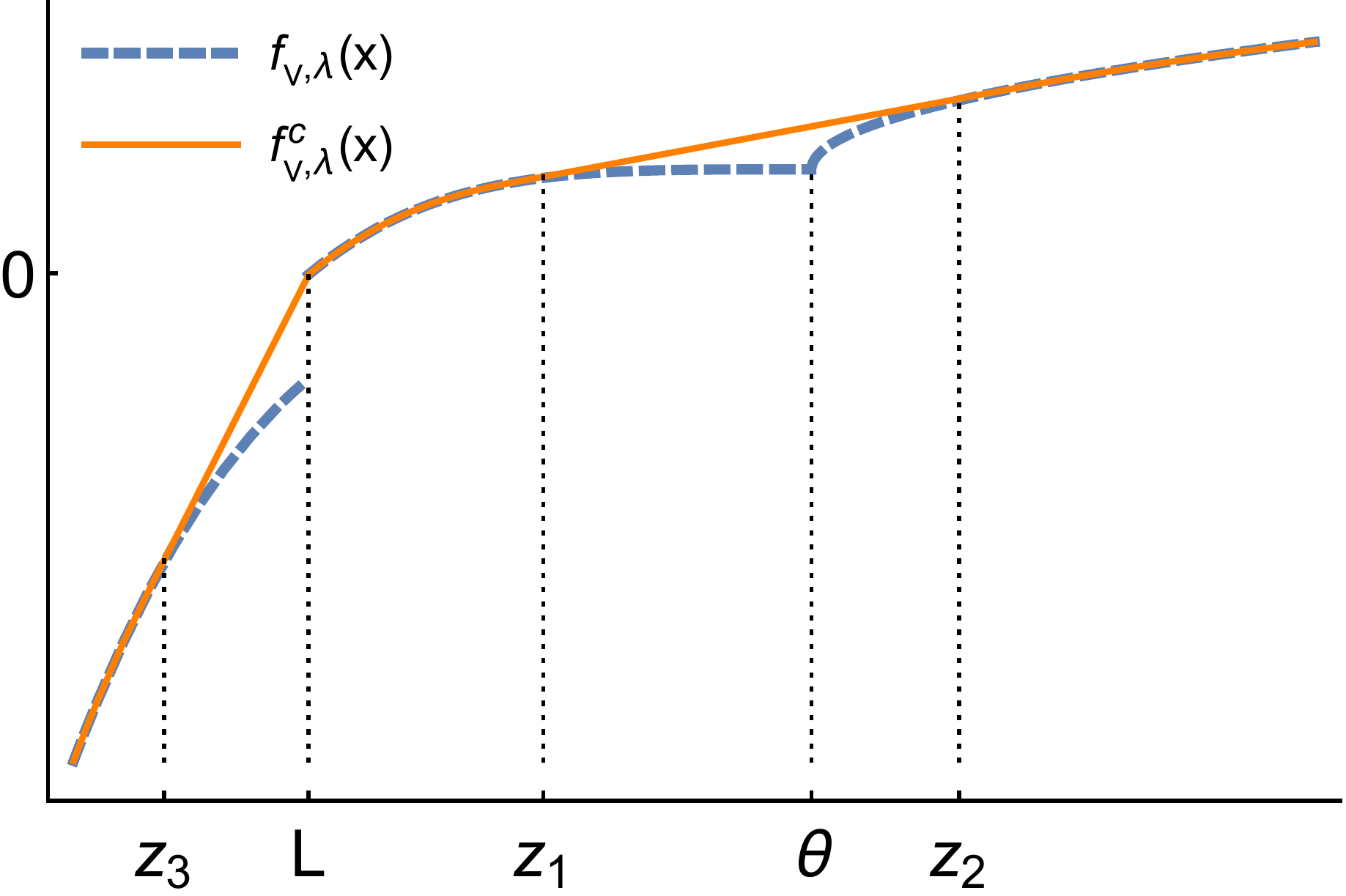}
		\caption{Schematic diagram of function $ f_{\nu,\lambda}(x) $ and its concave envelope function.}
		\label{fig:xfig2}
	\end{figure}

	And the optimal solution is\begin{equation}\nonumber x_{\nu,\lambda}^*(y)=\left\{
		\begin{aligned}
			& I_1(y)+\theta,&0<&y<U'(z_2-\theta),\\
			&\theta-I_2(\frac{y}{\nu}),&\nu D'(\theta-z_1)\leqslant &y<\nu D'(\theta-L),\\
			& L,&\nu D'(\theta-L)\leqslant &y<\nu D'(\theta-z_3),\\
			&\theta-I_2(\frac{y}{\nu}),&\nu D'(\theta-z_3)\leqslant &y<\nu D'(\theta),\\
			& 0,& \nu D'(\theta)\leqslant &y.
		\end{aligned}\right.
	\end{equation}
	We call this {\bf Case IV }($ L<z_1 $ and $ k_{\nu,\lambda}\leqslant \nu D'(\theta) $).
\end{enumerate}
\vskip 5pt
When $ z_1\leqslant L<\theta $, there exists a unique $ z_4>\theta $ such that
\begin{equation}\nonumber
	U'(z_4-\theta)=\frac{U(z_4-\theta)+\nu D(\theta-L)}{z_4-L}.
\end{equation}We also need to consider the relationship between $ k_{\nu,\lambda} $ and $ \nu D'(\theta) $:
\begin{enumerate}
	\item 	 If $ k_{\nu,\lambda} > \nu D'(\theta) $,  the relationship between $ k_{\nu,\lambda} $ and $ U'(z_4-\theta)$ is required.\begin{enumerate}
		\item If $ k_{\nu,\lambda} > U'(z_4-\theta) $, the schematic diagram is as shown in Fig.~\ref{fig:xfig3}.
		\begin{figure}[htbp] 
			\centering
			\includegraphics[width=0.6\linewidth]{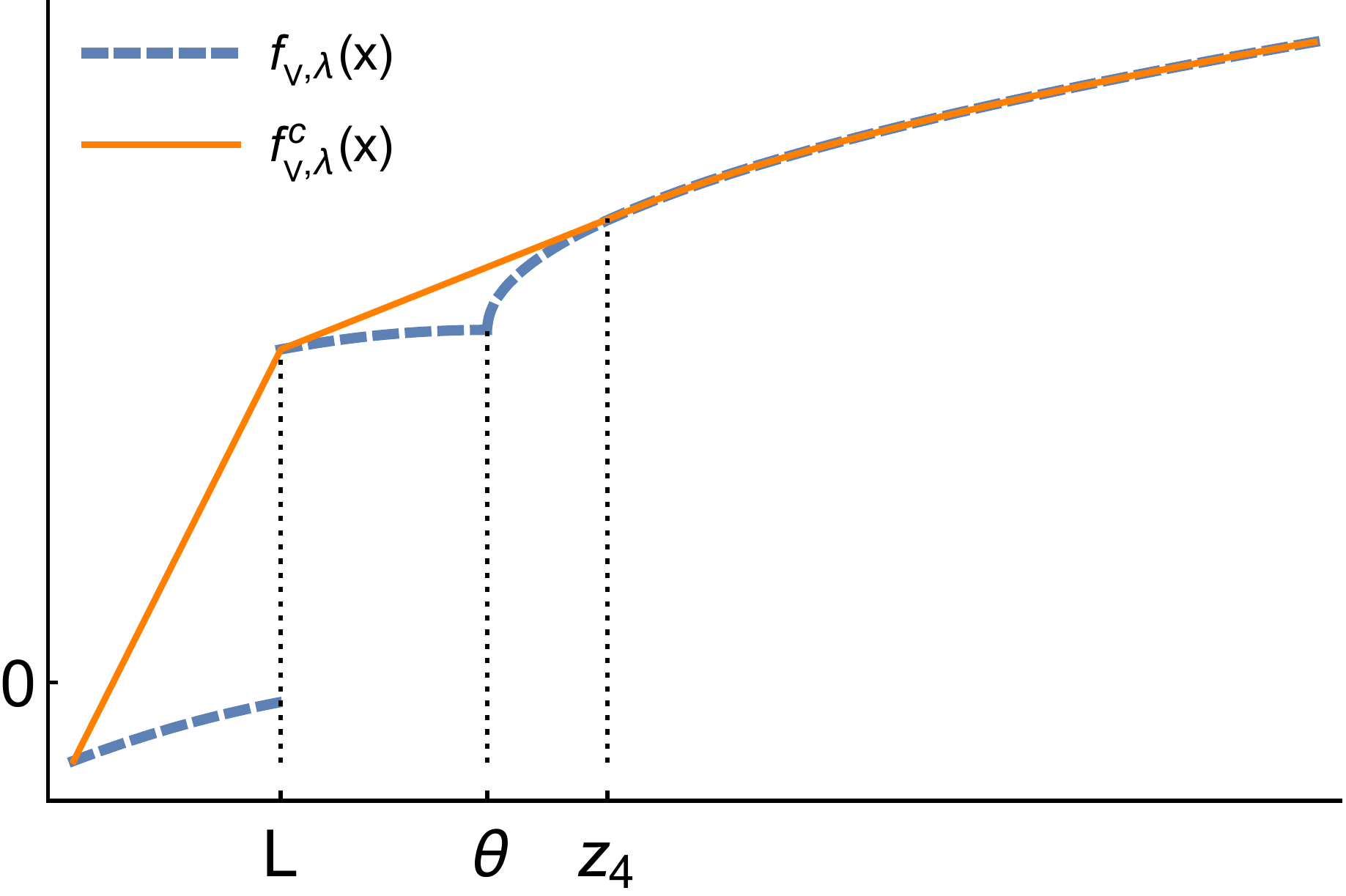}
			\caption{Schematic diagram of function $ f_{\nu,\lambda}(x) $ and its concave envelope function.}
			\label{fig:xfig3}
		\end{figure}
	
		The optimal solution is \begin{equation}\nonumber x_{\nu,\lambda}^*(y)=\left\{
			\begin{aligned}
				& I_1(y)+\theta,&0<&y<U'(z_4-\theta)\\
				& L,&U'(z_4-\theta)\leqslant &y<k_{\nu,\lambda},\\
				& 0,& k_{\nu,\lambda}\leqslant &y.
			\end{aligned}\right.
		\end{equation}
		We call this {\bf Case V }($ z_1\leqslant L<\theta $, $ k_{\nu,\lambda} > \nu D'(\theta) $ and  $ k_{\nu,\lambda} > U'(z_4-\theta) $).
		\item If $ k_{\nu,\lambda} \leqslant U'(z_4-\theta) $, there exists a unique $ z_5>\theta $ satisfying \begin{equation}\nonumber
			U'(z_5-\theta)=\frac{U(z_5-\theta)+\lambda+\nu D(\theta)}{z_5-0}.
		\end{equation}
		The schematic diagram is as shown in Fig.~\ref{fig:xfig4}.
		\begin{figure}[htbp] 
			\centering
			\includegraphics[width=0.6\linewidth]{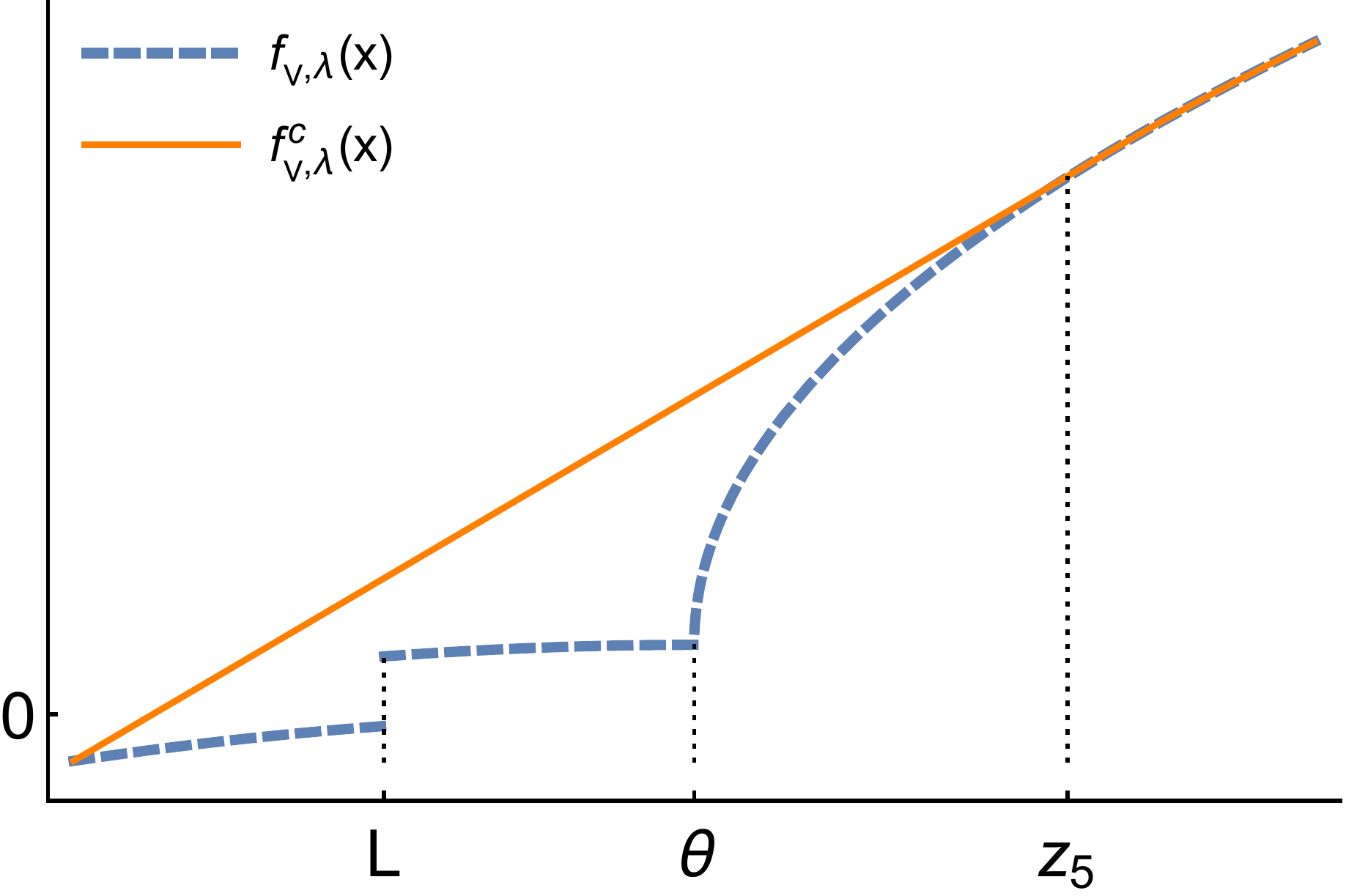}
			\caption{Schematic diagram of function $ f_{\nu,\lambda}(x) $ and its concave envelope function.}
			\label{fig:xfig4}
		\end{figure}
	
		The optimal solution is \begin{equation}\label{key4}x_{\nu,\lambda}^*(y)=\left\{
			\begin{aligned}
				& I_1(y)+\theta,&0<&y<U'(z_5-\theta),\\
				& 0,&U'(z_5-\theta)\leqslant &y.
			\end{aligned}\right.
		\end{equation}We call this {\bf Case VI} ($ z_1\leqslant L<\theta $, $ k_{\nu,\lambda} > \nu D'(\theta) $ and  $ k_{\nu,\lambda} \leqslant  U'(z_4-\theta) $).
	\end{enumerate}
	\item If $ k_{\nu,\lambda} \leqslant \nu D'(\theta) $, then there exists a unique $  z_3\in[0,L) $ satisfying Eq.~(\ref{z3}). We consider the relationship between $ \nu D'(\theta-z_3) $ and $  U'(z_4-\theta)  $. \begin{enumerate}
		\item If $ \nu D'(\theta-z_3) >  U'(z_4-\theta)  $, the schematic diagram is as shown in Fig.~\ref{fig:xfig5}.
		\begin{figure}[htbp] 
			\centering
			\includegraphics[width=0.6\linewidth]{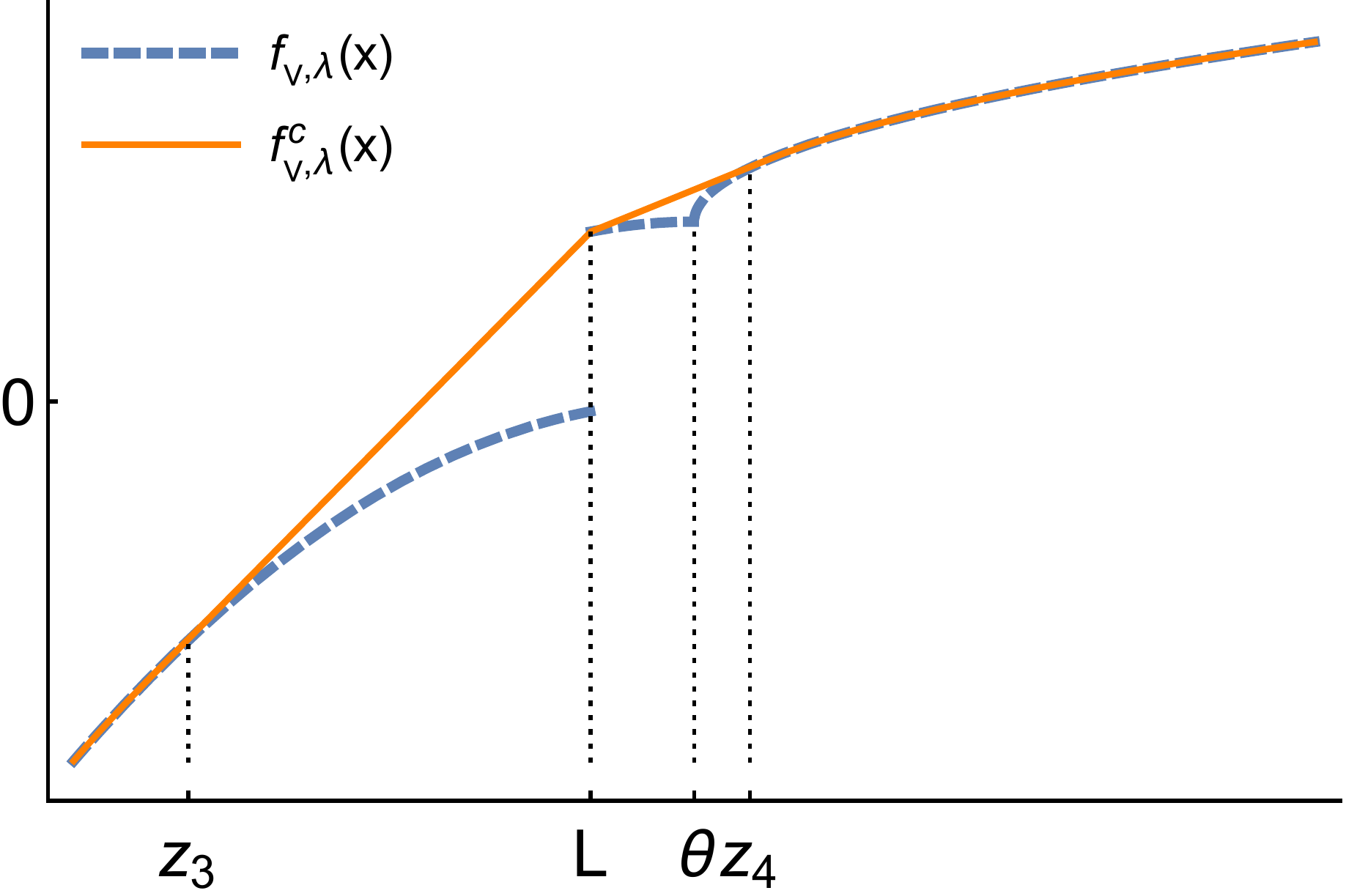}
			\caption{Schematic diagram of function $ f_{\nu,\lambda}(x) $ and its concave envelope function.}
			\label{fig:xfig5}
		\end{figure}
	
		The optimal solution in this case is \begin{equation}\nonumber x_{\nu,\lambda}^*(y)=\left\{
			\begin{aligned}
				& I_1(y)+\theta,&0<&y<U'(z_4-\theta),\\
				& L,&U'(z_4-\theta)\leqslant &y<\nu D'(\theta-z_3),\\
				&\theta-I_2(\frac{y}{\nu}),&\nu D'(\theta-z_3)\leqslant &y<\nu D'(\theta),\\
				& 0,&\nu D'(\theta)\leqslant &y.
			\end{aligned}\right.
		\end{equation}
\vskip 5pt
		We call this {\bf Case VII} ($ z_1\leqslant L<\theta $, $ k_{\nu,\lambda} \leqslant \nu D'(\theta) $ and  $ \nu D'(\theta-z_3) >  U'(z_4-\theta)  $).
		\item If $ \nu D'(\theta-z_3) \leqslant U'(z_4-\theta)  $, then there exist unique $ z_6$ and $z_7 $: $z_6<L<\theta<z_7 $ satisfying \begin{equation}\nonumber
			\nu D'(\theta-z_6)=U'(z_7-\theta)=\frac{U(z_7-\theta)+\lambda+\nu D(\theta-z_6)}{(z_7-\theta)+(\theta-z_6)}.
		\end{equation}
		If $ z_6\leqslant0 $, it belongs to the case of Fig.~\ref{fig:xfig4} and the optimal solution is given by Eq.~(\ref{key4}). If $0< z_6<L $, the schematic diagram is as shown in Fig.~\ref{fig:xfig6}.
		\begin{figure}[htbp] 
			\centering
			\includegraphics[width=0.6\linewidth]{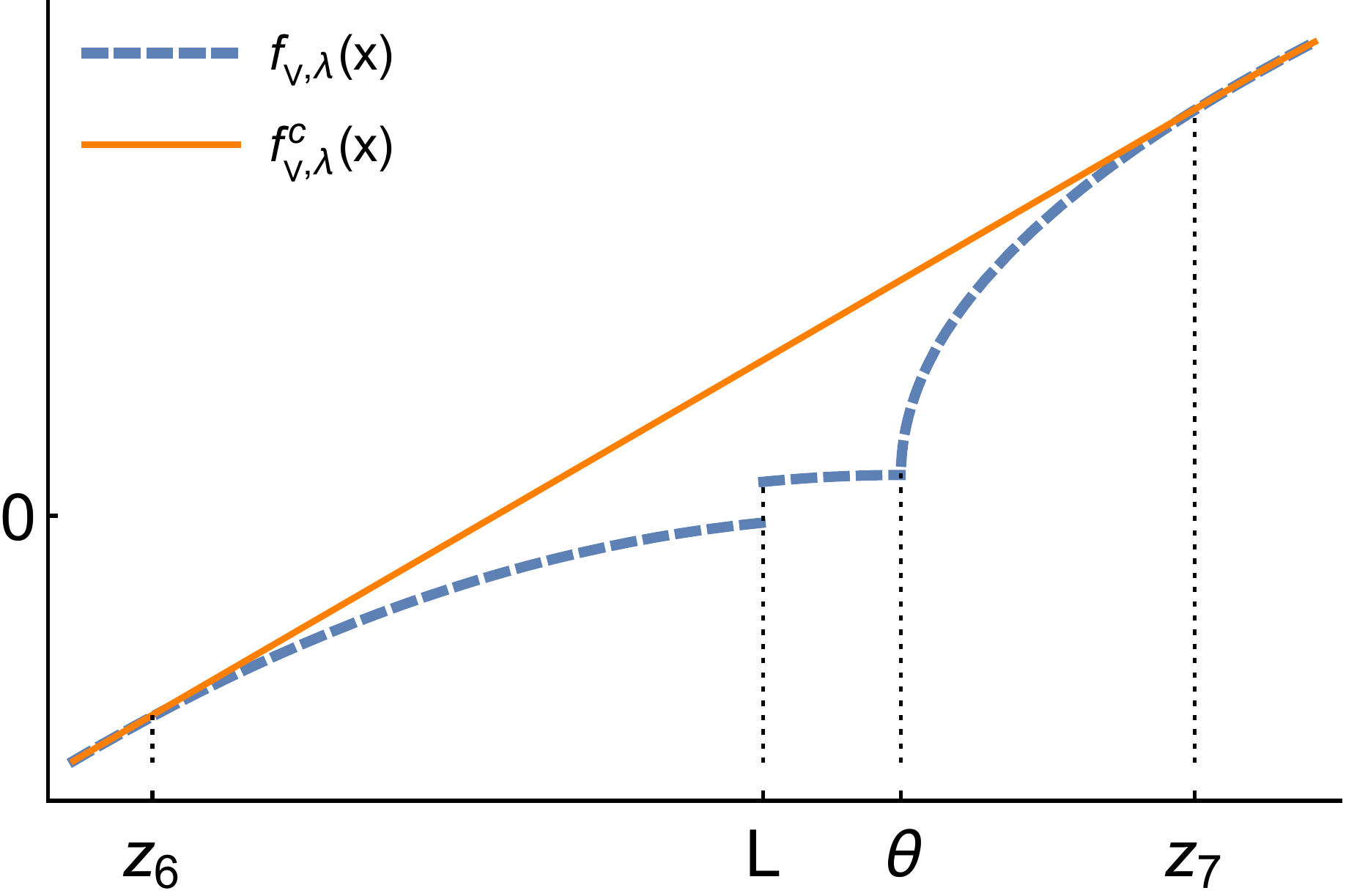}
			\caption{Schematic diagram of function $ f_{\nu,\lambda}(x) $ and its concave envelope function.}
			\label{fig:xfig6}
		\end{figure}
	
		The optimal solution is \begin{equation}\label{key6}x_{\nu,\lambda}^*(y)=\left\{
			\begin{aligned}
				& I_1(y)+\theta,&0<&y<U'(z_7-\theta),\\
				&\theta-I_2(\frac{y}{\nu}),&\nu D'(\theta-z_6)\leqslant &y<\nu D'(\theta),\\
				& 0,&\nu D'(\theta)\leqslant &y.
			\end{aligned}\right.
		\end{equation}
\vskip 5pt
We call this {\bf Case VIII} ($ z_1\leqslant L<\theta $, $ k_{\nu,\lambda} \leqslant \nu D'(\theta) $,  $ \nu D'(\theta-z_3) \leqslant  U'(z_4-\theta)  $ and $0< z_6<L $).
	\end{enumerate}	
\end{enumerate}
\vskip 5pt
When $ L=\theta $, similarly,  there are unique $ z_6'$ and $z_7' $:  $ z_6'<\theta<z_7' $ satisfying \begin{equation}\nonumber
	\nu D'(\theta-z_6')=U'(z_7'-\theta)=\frac{U(z_7'-\theta)+\lambda+\nu D(\theta-z_6')}{(z_7'-\theta)+(\theta-z_6')}.
\end{equation} If $ z_6'\leqslant0 $, it is reduced to the situation of Fig.~\ref{fig:xfig4} and the optimal solution is given by Eq.~(\ref{key4}). While if $0< z_6'<L $, it belongs to the case of Fig.~\ref{fig:xfig6} and Eq.~(\ref{key6}) is the optimal solution.

\vskip 5pt
When $ L>\theta $, there is a unique $ z_8\in[\theta,L)$ satisfying
\begin{equation}\nonumber
	U'(z_8-\theta)=\frac{U(L-\theta)+\lambda-U(z_8-\theta)}{L-z_8}=\frac{f_{\nu}(L)+\lambda-f_{\nu}(z_8)}{L-z_8}.
\end{equation}We also consider the relationship between $ k_{\nu,\lambda} $ and $ \nu D'(\theta) $: \begin{enumerate}
	\item If $ k_{\nu,\lambda} > \nu D'(\theta) $, we consider the relationships among $ U'(z_8-\theta)$, $U'(L-\theta)$ and $k_{\nu,\lambda} $ now.
\begin{enumerate}
		\item If $ U'(z_8-\theta)>k_{\nu,\lambda} \geqslant U'(L-\theta)$, the schematic diagram is as shown in Fig.~\ref{fig:xfig7}.
		\begin{figure}[htbp] 
			\centering
			\includegraphics[width=0.6\linewidth]{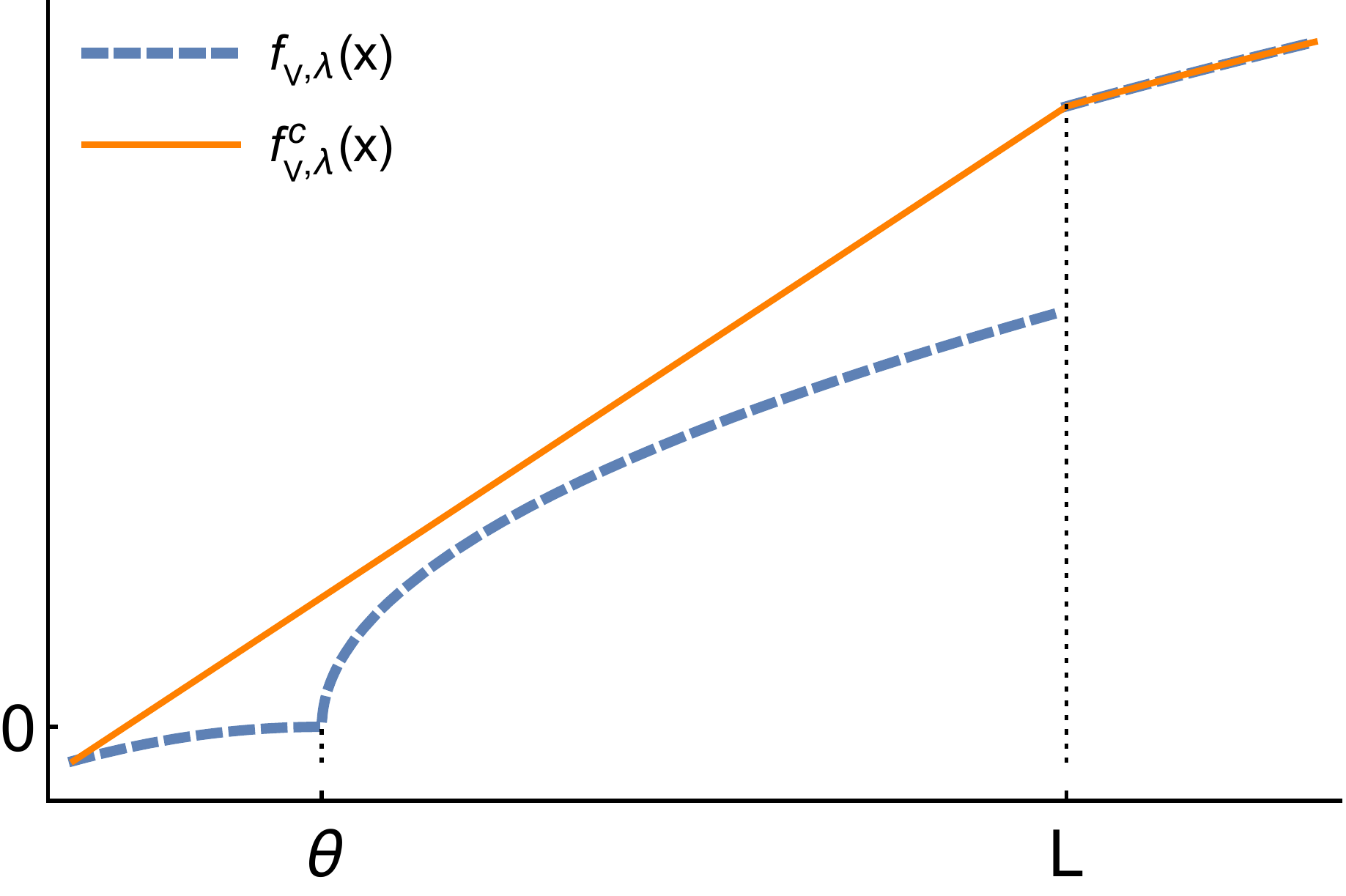}
			\caption{Schematic diagram of function $ f_{\nu,\lambda}(x) $ and its concave envelope function.}
			\label{fig:xfig7}
		\end{figure}
	
		The optimal solution in this case is \begin{equation}\nonumber x_{\nu,\lambda}^*(y)=\left\{
			\begin{aligned}
				& I_1(y)+\theta,&0<&y<U'(L-\theta),\\
				&L,&U'(L-\theta)\leqslant &y<k_{\nu,\lambda},\\
				& 0,&k_{\nu,\lambda}\leqslant &y.
			\end{aligned}\right.
		\end{equation}
\vskip 5pt
We call this {\bf Case IX }($ L>\theta $, $ k_{\nu,\lambda} > \nu D'(\theta) $ and $ U'(z_8-\theta)>k_{\nu,\lambda} \geqslant U'(L-\theta)$).
		\item If $ k_{\nu,\lambda} < U'(L-\theta)$, then there is a unique $ z_9>L$ satisfying
		\begin{equation}\nonumber
			U'(z_9-\theta)=\frac{U(z_9-\theta)+\lambda+\nu D(\theta)}{z_9-0}=\frac{f_{\nu}(z_9)+\lambda-f_{\nu}(0)}{z_9-0}.
		\end{equation}The schematic diagram is as shown in Fig.~\ref{fig:xfig8}.
		\begin{figure}[htbp] 
			\centering
			\includegraphics[width=0.6\linewidth]{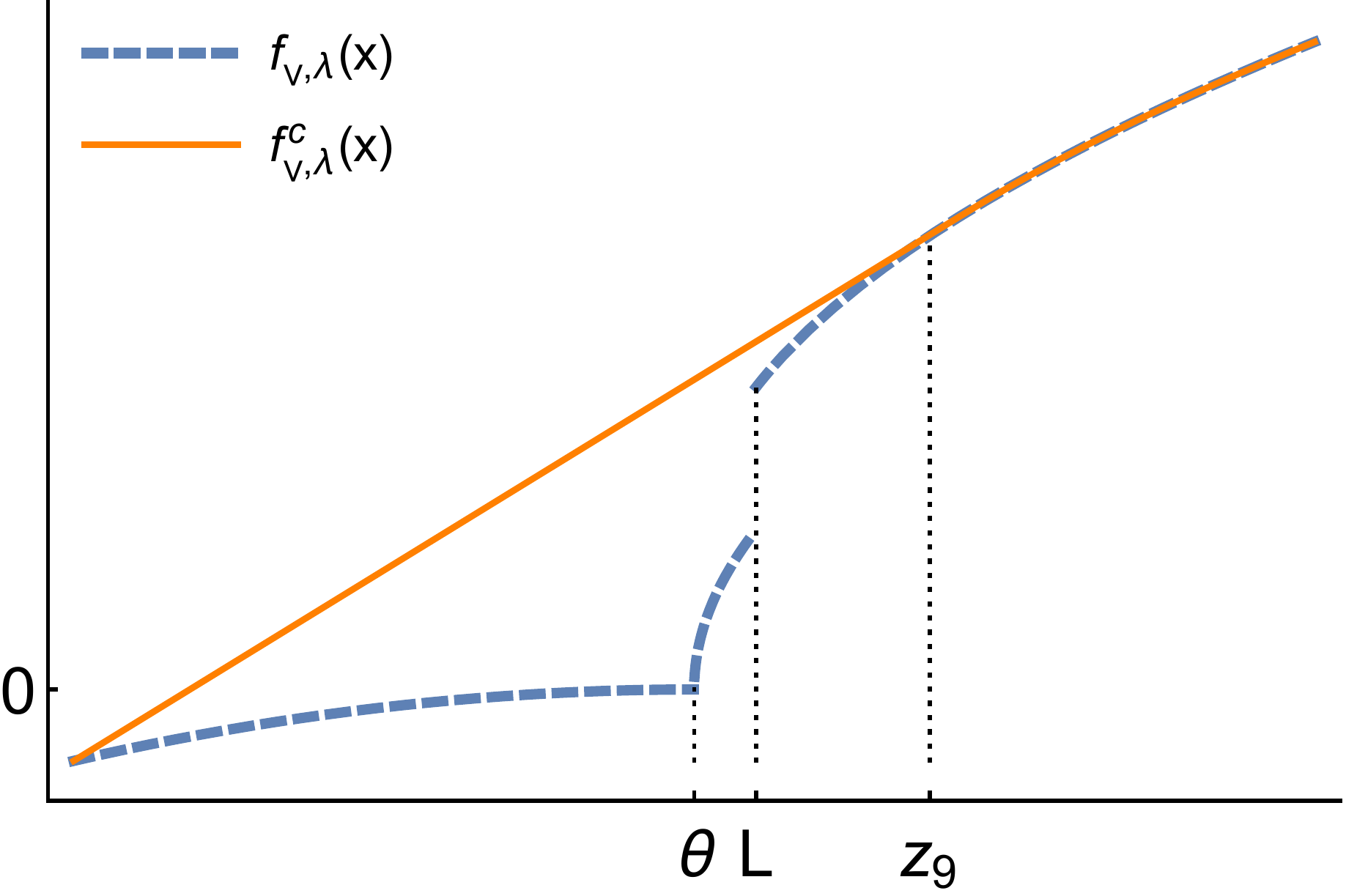}
			\caption{Schematic diagram of function $ f_{\nu,\lambda}(x) $ and its concave envelope function.}
			\label{fig:xfig8}
		\end{figure}
	
		The optimal solution in this case is \begin{equation}\label{key8}x_{\nu,\lambda}^*(y)=\left\{
			\begin{aligned}
				& I_1(y)+\theta,&0<&y<U'(z_9-\theta),\\
				& 0,&U'(z_9-\theta)\leqslant &y.
			\end{aligned}\right.
		\end{equation}
\vskip 5pt
We call this {\bf Case X} ($ L>\theta $, $ k_{\nu,\lambda} > \nu D'(\theta) $ and $ k_{\nu,\lambda} < U'(L-\theta)$).
		\item If $ k_{\nu,\lambda} \geqslant U'(z_8-\theta)$, then there exists a unique $ z_{10}\in(\theta,z_8]$ satisfying
		\begin{equation}\nonumber
			U'(z_{10}-\theta)=\frac{U(z_{10}-\theta)+\nu D(\theta)}{z_{10}-0}=\frac{f_{\nu}(z_{10})-f_{\nu}(0)}{z_{10}-0}.
		\end{equation}The schematic diagram is as shown in Fig.~\ref{fig:xfig9}.
		\begin{figure}[htbp] 
			\centering
			\includegraphics[width=0.6\linewidth]{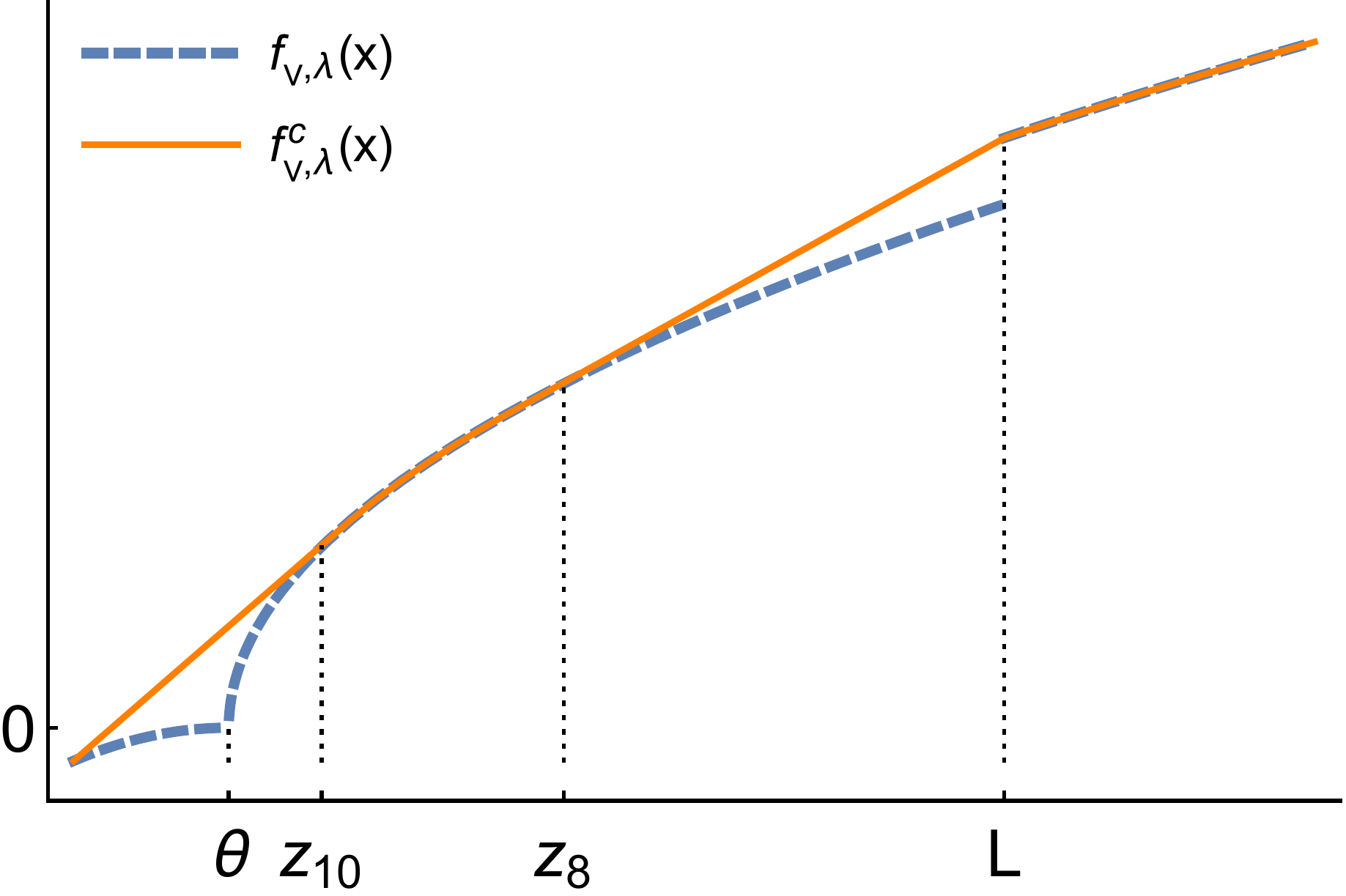}
			\caption{Schematic diagram of function $ f_{\nu,\lambda}(x) $ and its concave envelope function.}
			\label{fig:xfig9}
		\end{figure}
	
		The optimal solution in this case is \begin{equation}\nonumber x_{\nu,\lambda}^*(y)=\left\{
			\begin{aligned}
				& I_1(y)+\theta,&0<&y<U'(L-\theta),\\
				&L,&U'(L-\theta)\leqslant &y<U'(z_8-\theta),\\
				& I_1(y)+\theta,&U'(z_8-\theta)\leqslant&y<U'(z_{10}-\theta),\\
				& 0,&U'(z_{10}-\theta)\leqslant &y.
			\end{aligned}\right.
		\end{equation}
\vskip 5pt
We call  this {\bf Case XI} ($ L>\theta $, $ k_{\nu,\lambda} > \nu D'(\theta) $ and $ k_{\nu,\lambda} \geqslant U'(z_8-\theta)$).
	\end{enumerate}
	\item If $ k_{\nu,\lambda} \leqslant \nu D'(\theta) $, then there exists a unique $ z_{11}\in[0,\theta)$ satisfying
	\begin{equation}\nonumber
		\nu D'(\theta-z_{11})=\frac{U(L-\theta)+\lambda+\nu D(\theta-z_{11})}{L-z_{11}}.
	\end{equation}
We consider the relationships among $ U'(z_8-\theta)$, $U'(L-\theta)$ and $ \nu D'(\theta-z_{11}) $.\begin{enumerate}
		\item If $ U'(z_8-\theta)>\nu D'(\theta-z_{11}) \geqslant U'(L-\theta)$, the schematic diagram is as shown in Fig.~\ref{fig:xfig10}.
		\begin{figure}[htbp] 
			\centering
			\includegraphics[width=0.6\linewidth]{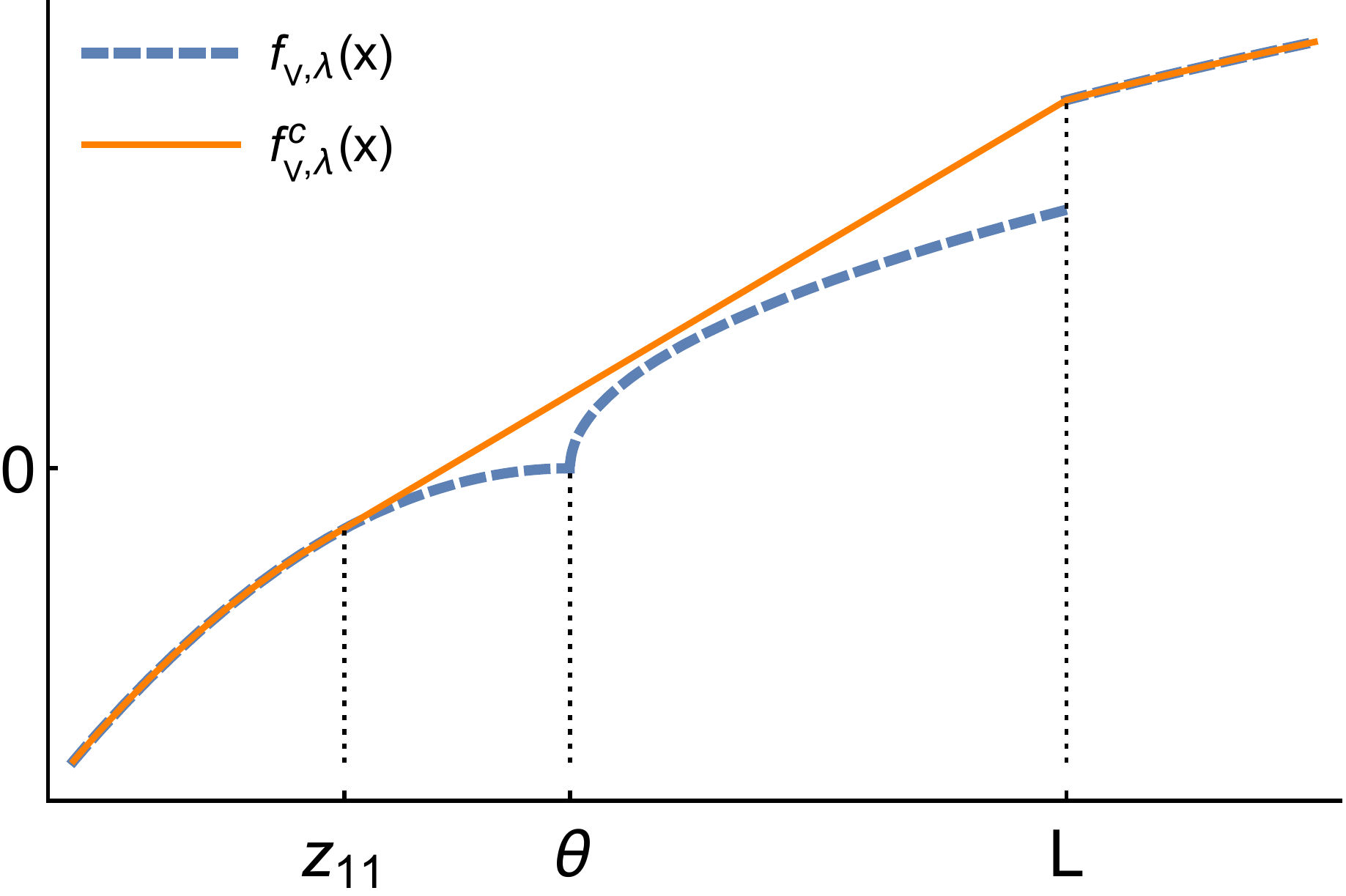}
			\caption{Schematic diagram of function $ f_{\nu,\lambda}(x) $ and its concave envelope function.}
			\label{fig:xfig10}
		\end{figure}
	
		The optimal solution in this case is  \begin{equation}\nonumber x_{\nu,\lambda}^*(y)=\left\{
			\begin{aligned}
				& I_1(y)+\theta,&0<&y<U'(L-\theta),\\
				&L,&U'(L-\theta)\leqslant &y<\nu D'(\theta-z_{11}),\\
				&\theta-I_2(\frac{y}{\nu}),&\nu D'(\theta-z_{11})\leqslant,&y<\nu D'(\theta),\\
				& 0,&\nu D'(\theta)\leqslant &y.
			\end{aligned}\right.
		\end{equation}	We call this {\bf Case XII} ($ L>\theta $, $ k_{\nu,\lambda} \leqslant \nu D'(\theta) $ and $ U'(z_8-\theta)>\nu D'(\theta-z_{11}) \geqslant U'(L-\theta)$).
		\item If $ U'(z_8-\theta)<\nu D'(\theta-z_{11})$ and $ z_2<z_8 $ hold, then the schematic diagram now is as shown in Fig.~\ref{fig:xfig11}.
		\begin{figure}[htbp] 
			\centering
			\includegraphics[width=0.6\linewidth]{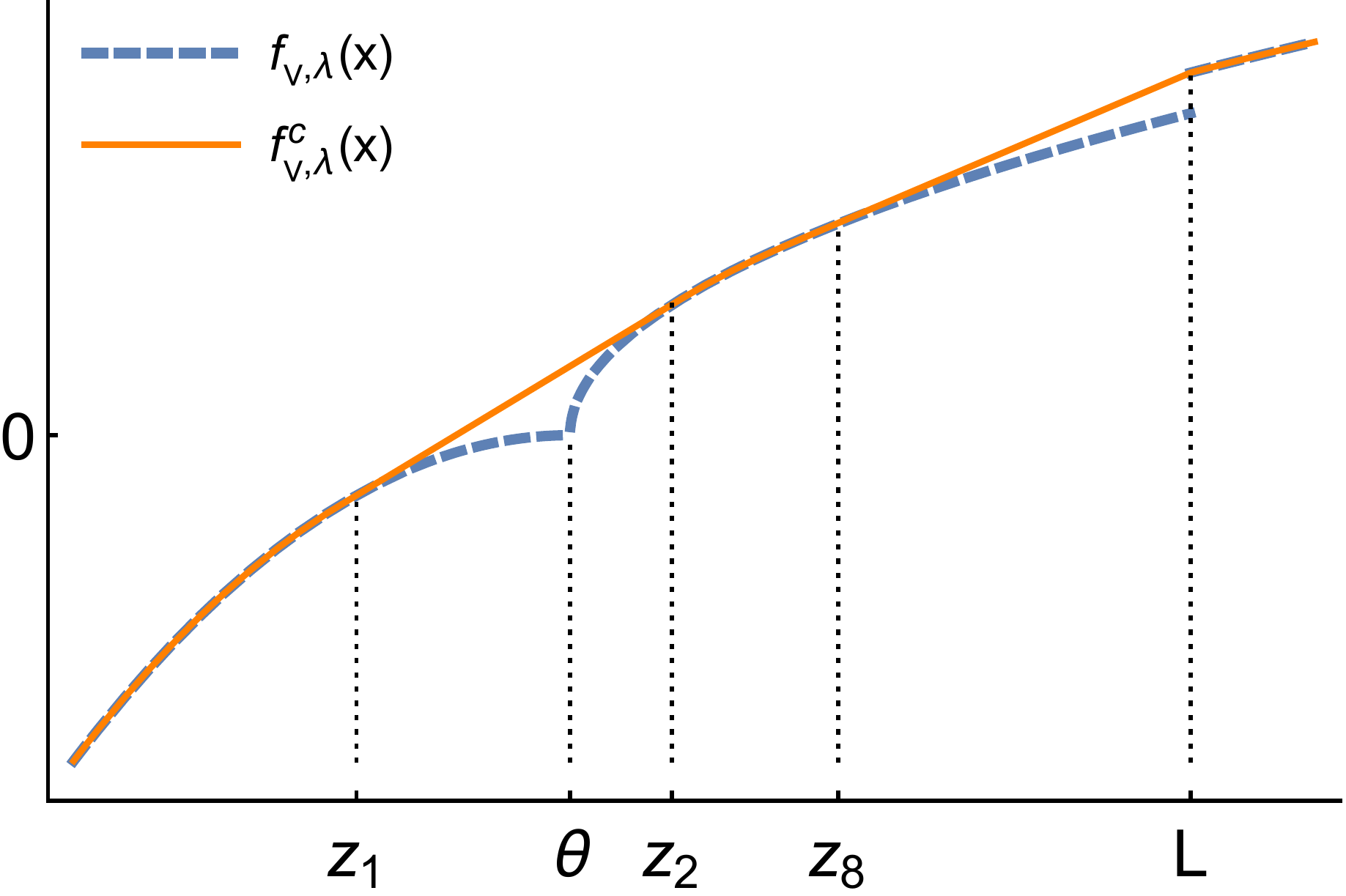}
			\caption{Schematic diagram of function $ f_{\nu,\lambda}(x) $ and its concave envelope function.}
			\label{fig:xfig11}
		\end{figure}
	
		The optimal solution in this case is \begin{equation}\nonumber x_{\nu,\lambda}^*(y)=\left\{
			\begin{aligned}
				& I_1(y)+\theta,&0<&y<U'(L-\theta),\\
				& L,&U'(L-\theta)\leqslant &y<U'(z_8-\theta),\\
				& I_1(y)+\theta,&U'(z_8-\theta)\leqslant&y<U'(z_2-\theta),\\
				&\theta-I_2(\frac{y}{\nu}),&\nu D'(\theta-z_1)\leqslant &y<\nu D'(\theta),\\
				& 0,& \nu D'(\theta)\leqslant &y.
			\end{aligned}\right.
		\end{equation}
\vskip 5pt
We call this {\bf Case XIII} ($ L>\theta $, $ k_{\nu,\lambda} \leqslant \nu D'(\theta) $ and $ U'(z_8-\theta)<\nu D'(\theta-z_{11})$).
		\item If $ U'(L-\theta)>\nu D'(\theta-z_{11}) $, then there exist unique $ z_{12}$ and $z_{13} $: $ z_{12}<\theta<L<z_{13} $ satisfying\begin{equation}\nonumber
			\nu D'(\theta-z_{12})=U'(z_{13}-\theta)=\frac{U(z_{13}-\theta)+\lambda+\nu D(\theta-z_{12})}{(z_{13}-\theta)+(\theta-z_{12})}.
		\end{equation}If $ z_{12}\leqslant0 $, it is reduced to the situation of Fig.~\ref{fig:xfig8} and the optimal solution is Eq.~(\ref{key8}). However, if $0< z_{12}<L $, the schematic diagram is as shown in Fig.~\ref{fig:xfig12}.
		\begin{figure}[htbp] 
			\centering
			\includegraphics[width=0.6\linewidth]{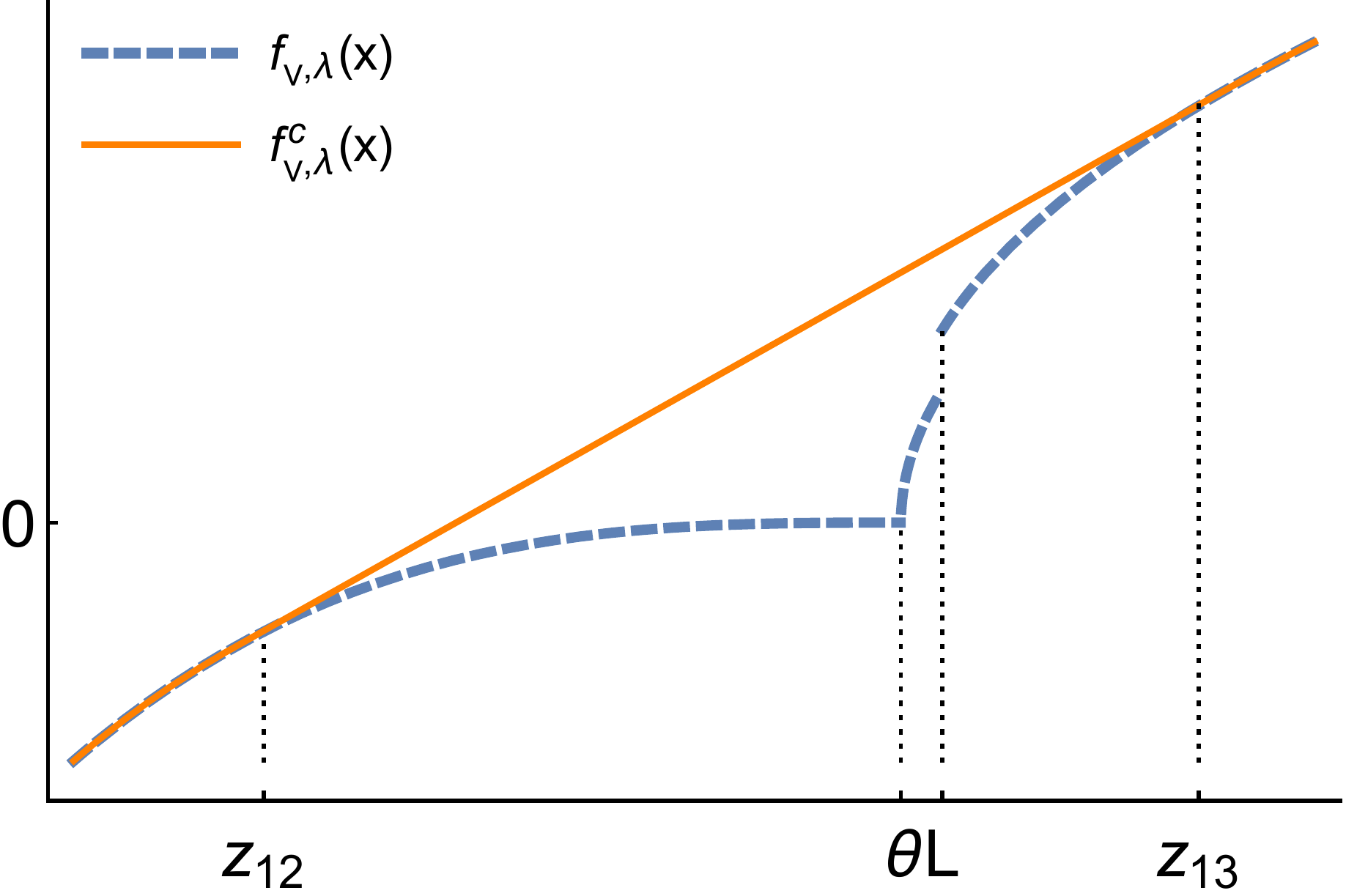}
			\caption{Schematic diagram of function $ f_{\nu,\lambda}(x) $ and its concave envelope function.}
			\label{fig:xfig12}
		\end{figure}
	
		The optimal solution in this case is \begin{equation}\nonumber x_{\nu,\lambda}^*(y)=\left\{
			\begin{aligned}
				& I_1(y)+\theta,&0<&y<U'(z_{13}-\theta),\\
				&\theta-I_2(\frac{y}{\nu}),&\nu D'(\theta-z_{12})\leqslant&y<\nu D'(\theta),\\
				& 0,&\nu D'(\theta)\leqslant &y.
			\end{aligned}\right.
		\end{equation}
\vskip 5pt
We call this {\bf Case XIV} ($ L>\theta $, $ k_{\nu,\lambda} \leqslant \nu D'(\theta) $, $ U'(L-\theta)>\nu D'(\theta-z_{11}) $ and $0< z_{12}<L $).
	\end{enumerate}
\end{enumerate}
\vskip 5pt
We see from the above paragraph that in the case of piece-wise concave function with VaR constraint, the optimal solution is very complicated and relies on the relationships among the reward function, penalty function and VaR constraint. There are in total fourteen cases. The optimal solution takes two-, three-, four- or five-region form depending on the relationships among $U(\cdot)$, $D(\cdot)$, $L$ and $\theta$.
\vskip 5pt
\subsubsection{\bf Concave penalty function}\label{ConcavePenaltyFunction}
When the penalty function is concave, the objective function of the optimization problem (\ref{LiVaR}) is the ``S-shaped" utility function which has been widely discussed in many work, see \cite{OIWS} for example. The manager is risk aversion towards gains while risk seeking towards losses. We only show the results and the proofs are omitted.
\vskip 5pt
In this case, there exists a unique solution $ z>\theta $ satisfying \begin{equation}\nonumber
	U(z-\theta)+\nu D(\theta)=zU'(z-\theta),
\end{equation}
where $ k_{\nu,\lambda}  $ is given in Eq.~(\ref{knulambda}).

\vskip 5pt
When $ L\geqslant z $,  we have $ k_{\nu,\lambda}\geqslant U'(L) $. And we consider the relationship between $ k_{\nu,\lambda} $ and $ U'(z) $.\begin{enumerate}
	\item[{\bf [Case I]}] If $ k_{\nu,\lambda}> U'(z-\theta) $, then we have \begin{equation}
		\nonumber x_{\nu,\lambda}^*(y)=\left\{
		\begin{aligned}
			& I_1(y)+\theta,&0<&y<U'(L-\theta),\\
			&L,&U'(L-\theta)\leqslant&y<k_{\nu,\lambda},\\
			& 0,&k_{\nu,\lambda}\leqslant &y.
		\end{aligned}\right.
	\end{equation}
	\item[{\bf [Case II]}] If $ U'(z)\geqslant k_{\nu,\lambda}\geqslant U'(L) $, then there exists a unique $ L_0\in[z,L] $ such that \begin{equation}\nonumber
		U(L-\theta)+\lambda-U(L_0-\theta)=(L-L_0)U'(L_0-\theta).
	\end{equation}
	The optimal solution is \begin{equation}
		\nonumber x_{\nu,\lambda}^*(y)=\left\{
		\begin{aligned}
			& I_1(y)+\theta,&0<&y<U'(L-\theta),\\
			&L,&U'(L-\theta)\leqslant,&y<U'(L_0-\theta),\\
			& I_1(y)+\theta,&U'(L_0-\theta)\leqslant&y<U'(z-\theta),\\
			& 0,&U'(z-\theta)\leqslant &y.
		\end{aligned}\right.
	\end{equation}
\end{enumerate}
\vskip 5pt
When $ z>L\geqslant \theta $, consider the relationship between $ k_{\nu,\lambda}  $ and $ U'(L-\theta) $.
\begin{enumerate}
	\item[{\bf [Case III]}] If $ k_{\nu,\lambda}\geqslant U'(L-\theta) $, then we have \begin{equation}
		\nonumber x_{\nu,\lambda}^*(y)=\left\{
		\begin{aligned}
			& I_1(y)+\theta,&0<&y<U'(L-\theta),\\
			&L,&U'(L-\theta)\leqslant&y<k_{\nu,\lambda},\\
			& 0,&k_{\nu,\lambda}\leqslant &y.
		\end{aligned}\right.
	\end{equation}
	\item[{\bf [Case IV]}] If $ k_{\nu,\lambda}< U'(L-\theta) $, then there exists a unique $ \tilde{z}_0\in(L,z] $ such that \begin{equation}\nonumber
		U(\tilde{z}_0-\theta)+\lambda+\nu D(\theta)=\tilde{z}_0U'(\tilde{z}_0-\theta).
	\end{equation}
	And we have \begin{equation}
		\nonumber x_{\nu,\lambda}^*(y)=\left\{
		\begin{aligned}
			& I_1(y)+\theta,&0<&y<U'(\tilde{z}_0-\theta),\\
			& 0,&U'(\tilde{z}_0-\theta)\leqslant &y.
		\end{aligned}\right.
	\end{equation}
\end{enumerate}
\paragraph{When $ L< \theta $}
Then there exists a unique $ \hat{z}\in (\theta, z) $ such that \begin{equation}\nonumber
	U(\hat{z}-\theta)+\nu D(\theta\-L)=(\hat{z}-L)U'(\hat{z}-\theta).
\end{equation}
\begin{enumerate}
	\item[{\bf [Case V]}] If $ k_{\nu,\lambda}>U'(\hat{z}-\theta) $, then we have \begin{equation}
		\nonumber x_{\nu,\lambda}^*(y)=\left\{
		\begin{aligned}
			& I_1(y)+\theta,&0<&y<U'(\hat{z}-\theta),\\
			&L,&U'(\hat{z}-\theta)\leqslant&y<k_{\nu,\lambda},\\
			& 0,&k_{\nu,\lambda}\leqslant &y.
		\end{aligned}\right.
	\end{equation}
	\item[{\bf [Case VI]}] If $ k_{\nu,\lambda}\leqslant U'(\hat{z}-\theta) $, then there exists a unique $ \hat{z}_0\in[\hat{z},z] $ such that \begin{equation}\nonumber
		U(\hat{z}_0-\theta)+\lambda+\nu D(\theta)=\tilde{z}_0U'(\hat{z}_0-\theta).
	\end{equation}
	The optimal solution is given by \begin{equation}
		\nonumber x_{\nu,\lambda}^*(y)=\left\{
		\begin{aligned}
			& I_1(y)+\theta,&0<&y<U'(\hat{z}_0-\theta),\\
			& 0,&U'(\hat{z}_0-\theta)\leqslant &y.
		\end{aligned}\right.
	\end{equation}
\end{enumerate}
\vskip 5pt
In the case of concave penalty function, the linearized problem has a utility function as in the prospect theory and there are in total six cases. Besides, the optimal solution takes two-, three- or four-region form depending on the relationships among the reward function, penalty function and the VaR constraint.
\vskip 10pt
\subsection{Summary}\label{conclusion}
We have presented the solution of the original problem (\ref{PSV}). The original problem contains a non-linear optimization goal and a VaR constraint. We first linearize the problem to an equivalent one with the expected utility goal under VaR constraint. Then, based on Lagrange dual method and concavification method, we obtain the optimal terminal wealth under different cases.  The reward function is concave and the optimal terminal wealth has different forms when the penalty function is convex or concave.
\vskip 5pt
Once the assumptions (H1)-(H4) hold and the current wealth value $ {X}^{\pi}(\cdot) $ is allowed to be negative. Using Theorem \ref{Theorem1}, Proposition \ref{Propos1}, Theorem \ref{Theorem2}, Corollary \ref{Corollary1}, Theorem \ref{VaR}, Theorem \ref{DeTr}, Theorem \ref{DET1}, Corollary \ref{Corollary2} and Eq.~(\ref{defined}), we can obtain the optimal wealth of the
 original optimization problem (\ref{PSV}) based on the following procedure:
%
\begin{enumerate}
	\item [Step 1:]Use the results  presented in Theorem \ref{DET1} to obtain the optimal solution $ x_{\nu,\lambda}^*(y) $ to the non-random problem (\ref{L}).
	\item [Step 2:]For the optimal solution $ x_{\nu,\lambda}^*(y) $ to non-random problem (\ref{L}), use Corollary \ref{Corollary2} to find the multiplier $ \beta^* =B_{\nu}(\lambda,\tilde{x}_{0})$ satisfying $ \mathbb{E}[H(T)Z_{\nu,\lambda,\beta^*}]=\tilde{x}_{0}  $\footnote{The $ \beta^* $ obtained here is implicit, but due to the monotonicity of $ R_{\nu}(\lambda,\beta) $, we can easily obtain  numerical solutions using numerical methods.}. Then the optimal solution of Problem (\ref{Li}) is $ Z_{\nu,\lambda} =Z_{\nu,\lambda,\beta^*}=x_{\nu,\lambda}^*(\beta^* H(T)) $.
	\item [Step 3:]Calculate $\displaystyle p\triangleq\sup\limits_{\lambda\geqslant0}S_{\nu}(\lambda,B_{\nu}(\lambda,\tilde{x}_{0})) $ to check whether $ p>1-\varepsilon $ or not. If not, the parameter should be adjusted. If $ p>1-\varepsilon $, find
	$ \lambda^*\geqslant0$\footnote{The $ \lambda^* $ obtained here is also implicit and is derived by numerical methods.} satisfying Eqs.~(\ref{TranCond1}) and (\ref{TranCond2}). Then $ Z_{\nu,\lambda^*} $ is the optimal solution to the linearized problem (\ref{LiVaR}).	
	\item [Step 4:]For the solution $ Z_{\nu} = Z_{\nu,\lambda^*} $ to the optimization problem (\ref{LiVaR}), find $ \nu^* $\footnote{The $ \nu^* $ obtained here is also implicit, and due to the dependence of $ \nu $ with $ \lambda^* $ and $ \beta^* $,  the numerical approach to find $ \nu^* $ is not an easy task.} to make the optimal value of the problem (\ref{LiVaR}) equal to 0. Then the corresponding solution $Z^*= Z_{\nu^*} $ is the solution to Problem (\ref{PVD}) and $ \nu^* $ is the optimal value of  Problems (\ref{PVD}), (\ref{PSVD}) and (\ref{PSV}).
	\item [Step 5:]By using Eqs.~(\ref{defined}), (\ref{Process1}) and  (\ref{Process2}) in Proposition \ref{Propos1}, the real terminal  wealth $X^{\pi^*}(T)/I(T)=\tilde{X}^{\pi^*}(T)=Z^* $ and optimal investment strategy $ \pi^* $  of Problem (\ref{PSV}) can be obtained theoretically.
\end{enumerate}
\vskip 10pt
\section{\bf Optimal investment strategies}
In the case that $U$ and $D$ are power functions, the optimal investment strategies have closed forms. Based on Proposition \ref{Propos1}, the optimal real wealth can be replicated by using the optimal terminal wealth:
%
%
\begin{equation}\nonumber
	\tilde{X}^{\pi^*}(t)=H(t)^{-1}
	\mathbb{E}[H(T)Z^*|\mathcal{F}_t],\quad 0\leqslant t\leqslant T,
\end{equation}
where $  Z^*=x_{\nu^*,\lambda^*}^*(\beta^* H(T))$.

Based on the results in Section \ref{4}, we know that $ Z^* $ can be expressed in the following form
\begin{equation}\label{equ:z*}
	Z^*=\sum_{i,k} \left[g^{(i)}I_k^{(i)}(\beta^*H(T))+b^{(i)}\right]\chi_{[h_1^{(i)},h_2^{(i)})}(\beta^*H(T)),\quad (k=1,2),
\end{equation}
where $g^{(i)}$, $ b^{(i)}$,  $h_1^{(i)}$ and $h_2^{(i)} $ are constants. $ I_k^{(i)}(\cdot),(k=1,2), $
are the inverse functions of $U'(x)$ and $D'(x)$, as such, they are also power functions. Without loss of generality, assume \begin{equation}\nonumber
	\tilde{X}^{\pi^*}(t) =\sum_{i,k}	H(t)^{-1}\mathbb{E}\left[H(T)[g^{(i)}H(T)^{\alpha^{(i)}}+b^{(i)}]\chi_{[h_1^{(i)},h_2^{(i)})}(\beta^*H(T))|\mathcal{F}_t\right].
\end{equation}
\vskip 5pt
 To calculate $ \tilde{X}^{\pi^*}(t) $, we only need to derive
\begin{equation}\nonumber
	H(t)^{-1}\mathbb{E}\left[H(T)[gH(T)^\alpha+b]\chi_{[h_1,h_2)}(\beta^*H(T))|\mathcal{F}_t\right].
\end{equation}
Noticing  Eq.~(\ref{DeFofH}), we know \begin{equation}\nonumber
	H(T)=H(t)\exp\left\{R(t)-\frac12\Sigma(t)+M(t)\right\},
\end{equation}
where \begin{equation}\nonumber
	\left\{	\begin{aligned}
		R(t)=&-\int_{t}^{T}r_r(s)\rd s,\\
		\Sigma(t)=&((\lambda_I-\sigma_{I})^2+\lambda_S^2)(T-t),\\
		M(t)=&(\lambda_I-\sigma_{I})(W_I(t)-W_I(T))+\lambda_S(W_S(t)-W_S(T)).
	\end{aligned}\right.
\end{equation}
$ H(t) $ is $ \mathcal{F}_t- $measurable,  $ M(t) $ is independent of  $ \mathcal{F}_t $ and $ M(t)\sim N(0,\Sigma(t)) $.
 As such, we have \begin{equation}\nonumber
	\begin{aligned}
		&H(t)^{-1}\mathbb{E}\left[H(T)gH(T)^\alpha\chi_{[h_1,h_2)}(\beta^*H(T))|\mathcal{F}_t\right]\\
		=&gH(t)^{-1}H(t)^{\alpha+1}\mathbb{E}\left[\exp\left((\alpha+1)R(t)-\frac{\alpha+1}{2}\Sigma(t)+(\alpha+1)M(t)\right)\cdot\right.\\&\qquad\qquad\qquad\qquad\left.\mathbf{1}_{\ln(\frac{h_1}{\beta^*H(t)})\leqslant R(t)-\frac12\Sigma(t)+M(t)<\ln(\frac{h_2}{\beta^*H(t)}) }\bigg|\mathcal{F}_t\right]\\
		=&gH(t)^{\alpha}\exp\left((\alpha+1)R(t)-\frac{\alpha+1}{2}\Sigma(t)\right)\\&\times\mathbb{E}\left[\exp\left((\alpha+1)\sqrt{\Sigma(t)}\cdot\frac{M(t)}{\sqrt{\Sigma(t)}}\right)\mathbf{1}_{H_1\leqslant \frac{M(t)}{\sqrt{\Sigma(t)}}<H_2 }\bigg|\mathcal{F}_t\right],
	\end{aligned}
\end{equation}
where \begin{equation}\nonumber
	H_1=\frac{\ln(\frac{h_1}{\beta^*H(t)})+\frac{1}{2}\Sigma(t)-R(t)}{\sqrt{\Sigma(t)}},\quad H_2=\frac{\ln(\frac{h_2}{\beta^*H(t)})+\frac{1}{2}\Sigma(t)-R(t)}{\sqrt{\Sigma(t)}}
\end{equation}are $\mathcal{F}_t- $measurable.
\vskip 5pt
Because \begin{equation}\nonumber
	\begin{aligned}
		&\mathbb{E}\left[\exp\left((\alpha+1)\sqrt{\Sigma(t)}\cdot\frac{M(t)}{\sqrt{\Sigma(t)}}\right)\mathbf{1}_{H_1\leqslant \frac{M(t)}{\sqrt{\Sigma(t)}}<H_2 }\bigg|\mathcal{F}_t\right]\\
		=&\int_{H_1}^{H_2}\exp((\alpha+1)\sqrt{\Sigma(t)}x)\frac{1}{\sqrt{2\pi}}\exp(-\frac{x^2}{2})\rd x\\
		=&\exp\left(\frac12(\alpha+1)^2\Sigma(t)
\right)\left[\Phi\left(H_2-(\alpha+1)\sqrt{\Sigma(t)}\right)
-\Phi\left(H_1-(\alpha+1)\sqrt{\Sigma(t)}\right)\right],
	\end{aligned}
\end{equation}
we obtain \begin{equation}\nonumber
	\begin{aligned}
		&H(t)^{-1}\mathbb{E}\left[H(T)[gH(T)^\alpha+b]\chi_{[h_1,h_2)}(\beta^*H(T))|\mathcal{F}_t\right]\\
		=&gH(t)^{\alpha}\exp\left((\alpha+1)R(t)+\frac12\alpha(\alpha+1)\Sigma(t)\right)\\&\times	\left[\Phi\left(H_2-(\alpha+1)\sqrt{\Sigma(t)}\right)-\Phi\left(H_1-(\alpha+1)\sqrt{\Sigma(t)}\right)\right]\\
		&+b\exp\left(R(t)\right)\times	\left[\Phi\left(H_2-\sqrt{\Sigma(t)}\right)-\Phi\left(H_1-\sqrt{\Sigma(t)}\right)\right].
	\end{aligned}
\end{equation}
Thus, the optimal real wealth at time $t$ can be expressed by
 \begin{equation}\nonumber
	\begin{aligned}
		\tilde{X}^{\pi^*}(t)
		=&\sum_{i,k}\left\{g^{(i)}H(t)^{\alpha^{(i)}}\exp\left((\alpha^{(i)}+1)R(t)+\frac12\alpha^{(i)}(\alpha^{(i)}+1)\Sigma(t)\right)\right.\\&\times	\left[\Phi\left(H_2^{(i)}-(\alpha^{(i)}+1)\sqrt{\Sigma(t)}\right)-\Phi\left(H_1^{(i)}-(\alpha^{(i)}+1)\sqrt{\Sigma(t)}\right)\right]\\
		&\left.+b^{(i)}\exp\left(R(t)\right)\times	\left[\Phi\left(H_2^{(i)}-\sqrt{\Sigma(t)}\right)-\Phi\left(H_1^{(i)}-\sqrt{\Sigma(t)}\right)\right]\right\}.
	\end{aligned}
\end{equation}
Using martingale method, we can derive the optimal strategies. We need to obtain the differential of $\tilde{X}^{\pi^*}$ first.  Define \[ \Psi(\alpha,b,g,h_1,h_2)\triangleq\frac{\partial}{\partial H(t)}\left(H(t)^{-1}\mathbb{E}\left[H(T)[gH(T)^\alpha+b]\chi_{[h_1,h_2)}(\beta^*H(T))|\mathcal{F}_t\right]\right)\cdot (-H(t)).\]
Then we have
\begin{equation}\nonumber
	\begin{split}&\Psi(\alpha,b,g,h_1,h_2)\\
		=&\frac{\partial}{\partial H(t)}\left(H(t)^{-1}\mathbb{E}\left[H(T)[gH(T)^\alpha+b]\chi_{[h_1,h_2)}(\beta^*H(T))|\mathcal{F}_t\right]\right)\cdot (-H(t))\\
		=&gH(t)^{\alpha}\exp\left((\alpha+1)R(t)+\frac{1}{2}\alpha(\alpha+1)\Sigma(t)
		\right)
		\\&\times\left\{	\alpha\left[\Phi\left(H_1-(\alpha+1)\sqrt{\Sigma(t)}\right)-\Phi\left(H_2-(\alpha+1)\sqrt{\Sigma(t)}\right)\right]\right.
\\&\left.+\frac{1}{\sqrt{\Sigma(t)}}\left[\phi\left(H_2-(\alpha+1)\sqrt{\Sigma(t)}\right)-\phi\left(H_1-(\alpha+1)\sqrt{\Sigma(t)}\right)\right]
\right\}\\
		&+b\exp\left(R(t)\right)\times	\frac{1}{\sqrt{\Sigma(t)}}\left[\phi\left(H_2-\sqrt{\Sigma(t)}\right)-\phi\left(H_1-\sqrt{\Sigma(t)}\right)\right].
	\end{split}
\end{equation}
As such, we obtain the differential of $\tilde{X}^{\pi^*}$:
\begin{equation}\label{dt}
	\rd \tilde{X}^{\pi^*}(t)=\square \rd t+\sum_{i,k}\Psi(\alpha^{(i)},b^{(i)},g^{(i)},h_1^{(i)},h_2^{(i)})\left[(\lambda_I-\sigma_I)\rd W_I(t)+\lambda_S\rd W_S(t)\right].
\end{equation}
In order to derive the optimal strategies, we need to compare the diffusion terms. As such, we only calculate the diffusion term of  $\tilde{X}^{\pi^*}$.

Comparing Eqs.~(\ref{dt}) and (\ref{SDE2}), we obtain the optimal investment strategies as follows: \begin{equation}\nonumber\left\{\begin{aligned}
		\pi_S(t)=\frac{1}{\sigma_{S_2}}&\left[\lambda_SI(t)\sum_{i,k}\Psi(\alpha^{(i)},b^{(i)},g^{(i)},h_1^{(i)},h_2^{(i)})-\sigma_{C_2}F(t)\right],\\
		\pi_P(t)=\frac{1}{\sigma_{S_1}}&\left[(\lambda_I-\sigma_I)I(t)\sum_{i,k}\Psi(\alpha^{(i)},b^{(i)},g^{(i)},h_1^{(i)},h_2^{(i)})\right.\\
		&\left.+\sigma_II(t)\tilde{X}^{\pi^*}(t)-\sigma_{S_1}\pi_S(t)-\sigma_{C_1}F(t)\right].
	\end{aligned}\right.
\end{equation}
The optimal terminal wealth is presented in Subsection \ref{ss:o} under different cases.  Subsection \ref{ss:o} shows that the optimal terminal wealth can be expressed by Eq.~(\ref{equ:z*}). Combining with the last equation, we have the closed forms of the optimal investment strategies.
\vskip 10pt
\section{\bf Sensitivity analysis}
In this section, we present numerical results of the pension manager under performance ratio and VaR constraint. Unless otherwise stated, the parameters we adopt are: $ T=40$, $ r_n=0.04$, $ r_r=0.02$, $  \sigma_{I}=0.4 $, $\sigma_{C_1}=0.2$, $\sigma_{C_2}=0.3$, $\sigma_{S_1}=0.3$, $\sigma_{S_2}=0.4$, $ \mu=0.1$, $ \lambda_I=0.2$, $\lambda_S=0.3$, $ i_0=1$, $ c_0=0.8$, $ x_0=1$. Furthermore, suppose that the reward and penalty functions are both power functions:  $ U(x)=x^{\gamma_1} $ and $D(x)= Ax^{\gamma_2}$, $ A=1$, $\gamma_1=0.3$, $\gamma_2=2.2 $. In this case, the reward function is concave while the penalty function is convex. Then the inverse function of $U'(x)$ is $I_1(x)=(\frac{x}{\gamma_1})^{\frac{1}{\gamma_1-1}}$, and the inverse function of $D'(x)$ is $I_2(x)=(\frac{x}{A\gamma_2})^{\frac{1}{\gamma_2-1}}$.
\vskip 5pt

%
\vskip 5pt
We are interested in the impact of the reference level $\theta$ on the optimal wealth. First set $ \varepsilon=0.01, L=6.5$ to ensure the condition  (\ref{condition}).
Fig.~\ref{theta} depicts the evolutions of  $ x_{\nu^*,\lambda^*}^*(\beta^*s) $ w.r.t. $s$ for $ \theta=6,~6.5,~7 $, respectively.  $ \theta=6,~6.5,~7 $ belongs to Case IV, Case IV and Case III, respectively. Fig.~\ref{theta} means the evolution of  $ Z^* $ versus the pricing kernel $ H(T) $. As shown in Fig.~\ref{theta}, the optimal wealth $Z^*$ takes four- or five-region form for different $\theta$. In these three cases, the bad-states regions are almost the same. When $\theta$ increases, the optimal wealth increases first and then decreases. In the good-states region,  the optimal wealth  has a positive relationship with the reference point. Besides, when $\theta$ increases, the intermediate-states region shrinks and the optimal wealth in this region increases.


\begin{figure}[htbp] 
	\centering
	\includegraphics[width=0.7\linewidth]{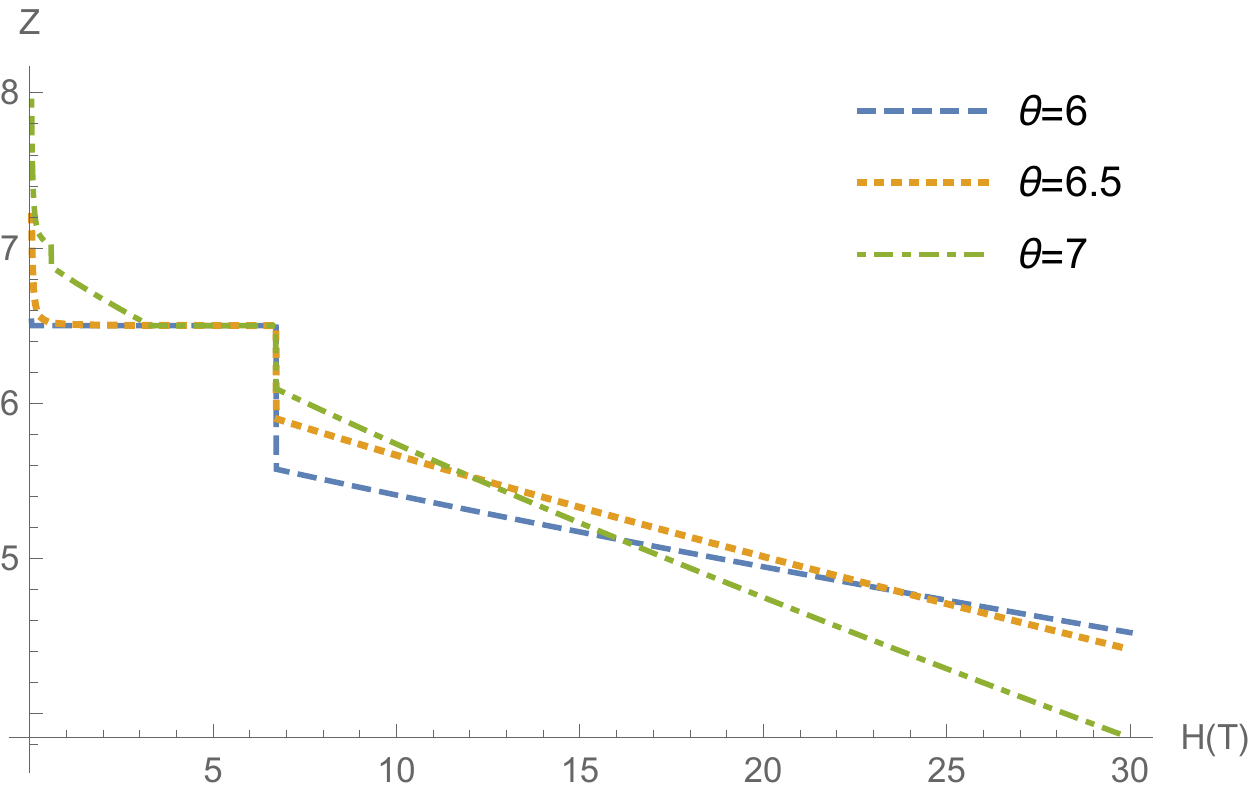}
	\caption{Graph of $ Z^* $ versus $ H(T) $ for different $ \theta\ $ when $ L=6.5$, $\varepsilon =0.01 $. }
	\label{theta}
\end{figure}
\vskip 5pt
We also vary $ \gamma_1$ and $\gamma_2 $ to illustrate the impacts of $ \gamma_1$ and $\gamma_2 $ on the optimal wealth. We set $ L=6.5$, $\theta =7 $ in the VaR constraint. For $\gamma_1$, we consider  $ \gamma_1=0.25,~0.3,~0.35 $ and $ \gamma_2=2.2 $, which is shown in Fig.~\ref{gamma1}. For $\gamma_2$, we choose $ \gamma_1=0.3 $ and $ \gamma_2=2,~2.2,~2.25 $, which is depicted  in Fig.~\ref{gamma2}.
\begin{figure}[htbp] 
	\centering
	\includegraphics[width=0.7\linewidth]{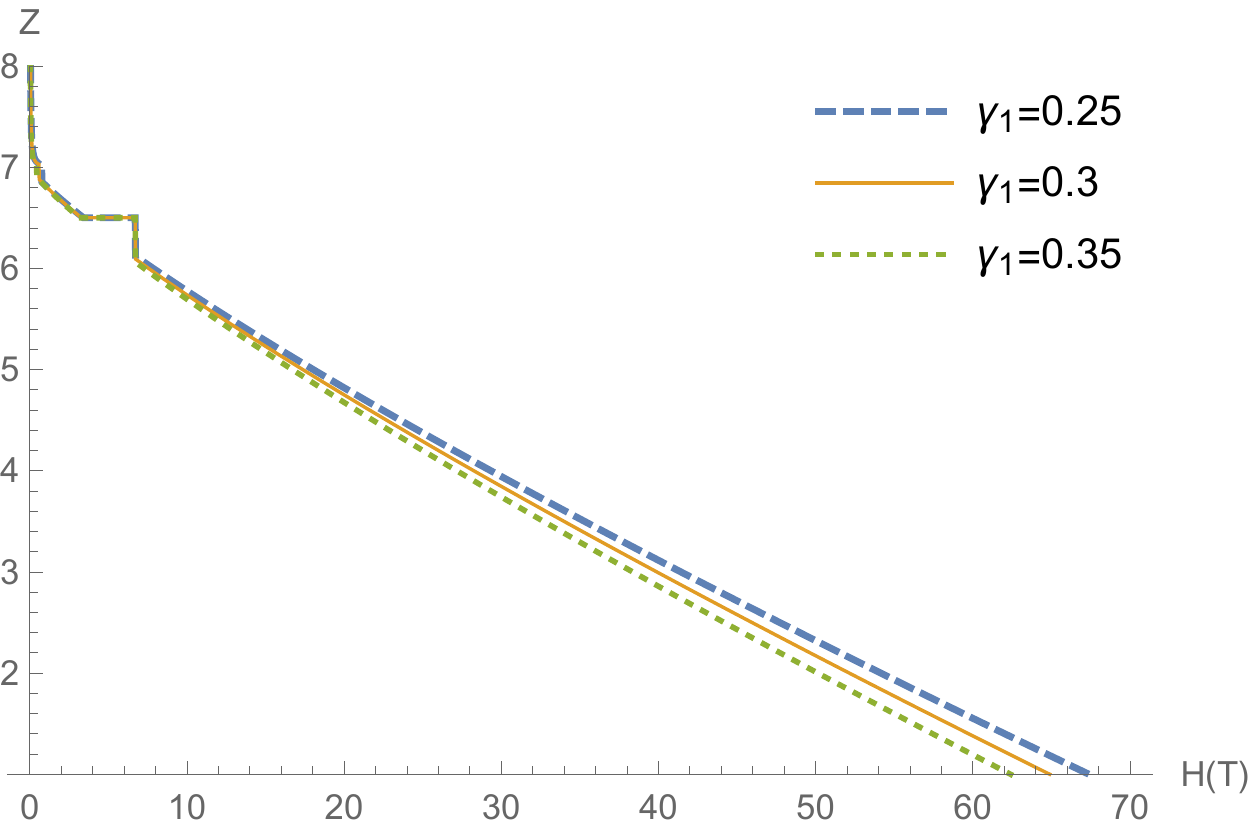}
	\caption{Graph of $ Z^* $ versus $ H(T) $ under different $ \gamma_1$. }
	\label{gamma1}
\end{figure}
In Fig.~\ref{gamma1}, we see that in the good-states region, the optimal wealth is almost the same for different $\gamma_1$. $\gamma_1$ characterizes the manager's risk attitude towards gains and larger $\gamma_1$ means less risk aversion. In the bad-states region, the optimal wealth decreases with $\gamma_1$, which means that the manager will give up some wealth in the bad-states region to ensure the wealth in the good-states region for larger $\gamma_1$.
\begin{figure}[htbp] 
	\centering
	\includegraphics[width=0.7\linewidth]{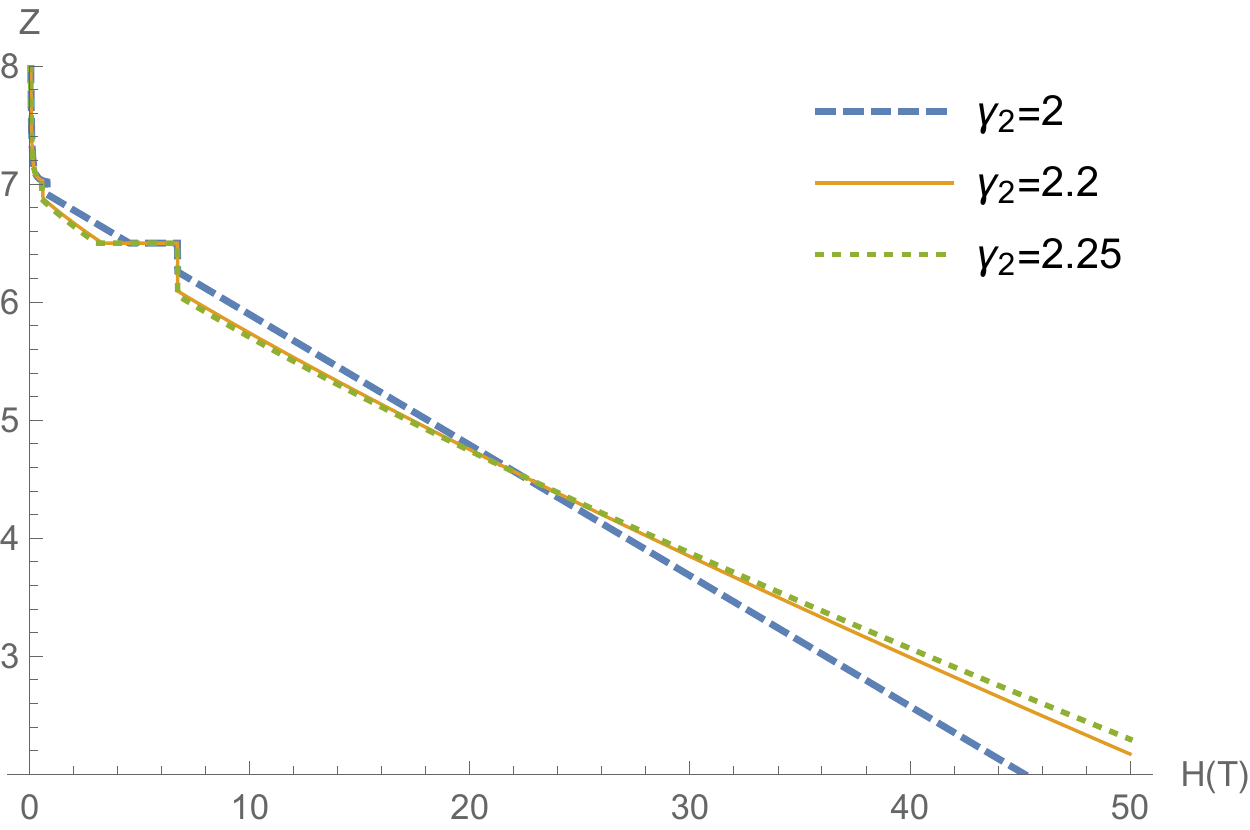}
	\caption{Graph of $ Z^* $ versus $ H(T) $ under different $ \gamma_2$. }
	\label{gamma2}
\end{figure}
$\gamma_2$ shows the degree of risk aversion attitude towards losses. When $\gamma_2$ increases, the manager is less risk averse towards the absolute wealth below $\theta$. We can observe from Fig.~\ref{gamma2} that the optimal wealth in the bad-states region increases with $\gamma_2$. Besides, in the good-states region, the optimal wealth decreases with $\gamma_2$.
\vskip 5pt
We are also concerned  with the impact of the VaR constraint $L$ on the optimal wealth. Set  $ \varepsilon=0.01, \theta=6$. Fig.~\ref{LL} reveals the evolutions of $Z^*$ for different  $ L=5,~5.5,~6,~6.5 $. The optimal wealth takes two- or three-region form according to the value of $L$. As the optimal solution when $ L=5.5 \text{~and~} \varepsilon=1$ satisfies the VaR constraint $ P(Z^*\geqslant L)\geqslant1-0.01 $, the VaR constraint is not binding. Similarly, when   $ L=5 $, the VaR constraint is also not binding. As such, the optimal solution for $L=5.5$ and $L=5$ are the same, which is illustrated in Fig.~\ref{LL}. We see that when $L$ increases from 6 to 6.5, the intermediate-states region enlarges and the optimal wealth in this region also increases. Meanwhile, the good-states region shrinks and the related optimal wealth increases. The bad-states region almost does not change. However, the optimal wealth in the bad-states region decreases with $L$.
\begin{figure}[htbp] 
	\centering
	\includegraphics[width=0.7\linewidth]{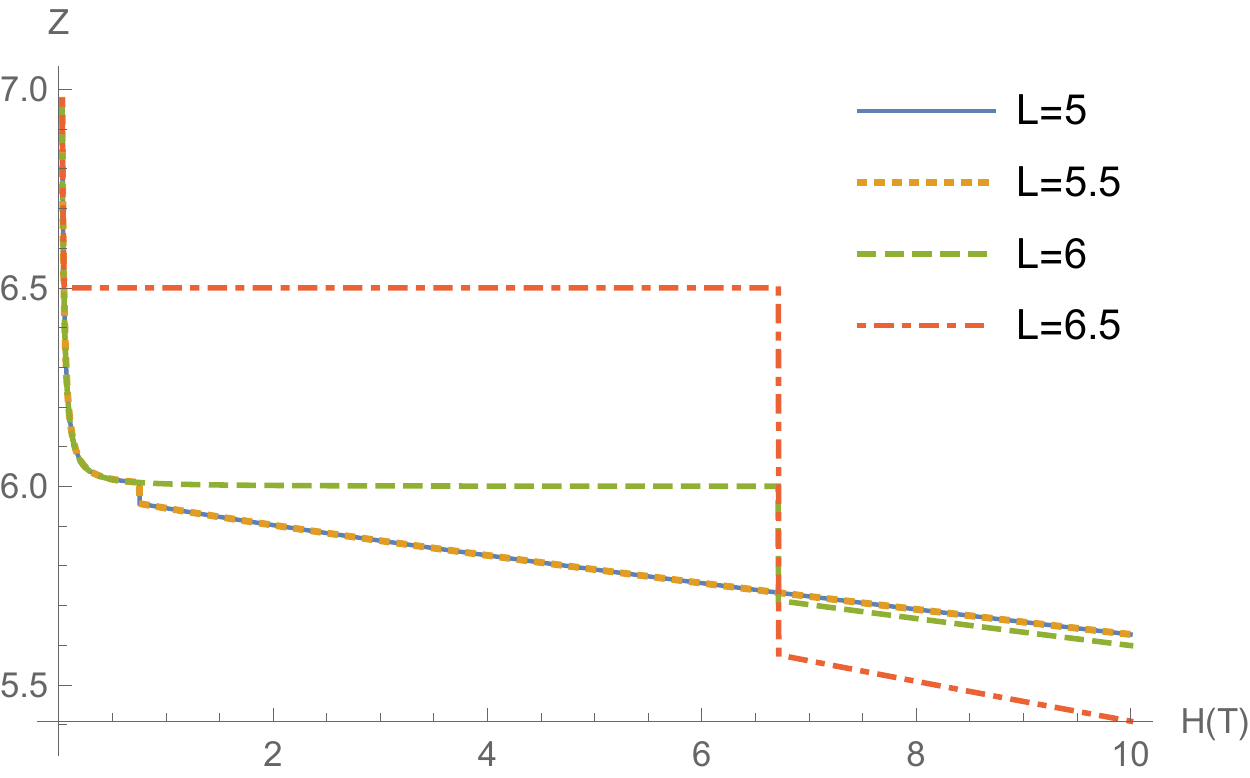}
	\caption{Graph of $ Z^* $ versus $ H(T) $ for different $ L\ $ when $\theta =6$, $\varepsilon=0.01$. }
	\label{LL}
\end{figure}
\vskip 5pt
To illustrate the impact of $ \varepsilon $ on the optimal wealth, we choose $ L=6$, $\theta =7.5 $ and consider $ \varepsilon=0.1,~0.01,~0.003$ and $0.001 $.  Fig.~\ref{vare} depicts the evolutions of optimal wealth corresponding to different $ \varepsilon$. The optimal wealth takes two- or three-region form depending on $ \varepsilon$. When $ \varepsilon$ decreases, the pension manager expects that the optimal terminal wealth higher than $L$ with a larger probability. When $\varepsilon=0.1,~0.01$, the VaR constraint is not binding and has no influence on the terminal wealth. As such, the optimal wealth coincides for $\varepsilon=0.1,~0.01$ in Fig.~\ref{vare}. When $ \varepsilon$ decreases from 0.003 to 0.001, the good-states region becomes bigger and the optimal wealth in the region also increases. Meanwhile, the intermediate-states region shrinks and the bad-states region enlarges. The optimal wealth in the  bad-states region increases with $\varepsilon$.
%
\begin{figure}[htbp] 
	\centering
	\includegraphics[width=0.7\linewidth]{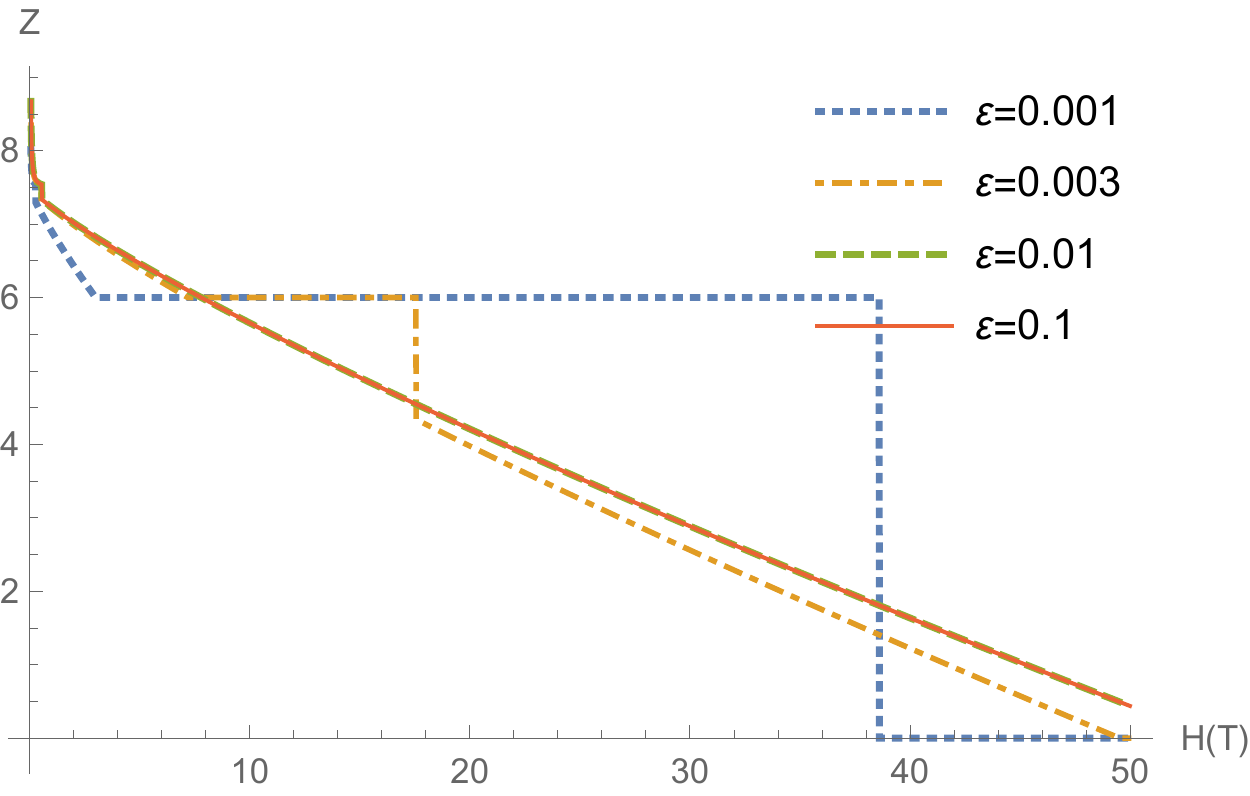}
	\caption{Graph of $ Z^* $ versus $ H(T) $ for different $ \varepsilon\ $ when $L=6,\theta =7.5 $. }
	\label{vare}
\end{figure}
\vskip 5pt
We also compare our results with the optimization goals in \cite{OIWS} and \cite{POWPR}. When $\varepsilon=1$, there is no VaR constraint and the optimization goal is the same as \cite{POWPR}. When $\varepsilon=0.01$, $ \nu=1$, the utility function is a piece-wise function  and has a similar form as in \cite{OIWS}. In Fig.~\ref{comp}, we see that the good-states region enlarges with VaR constraint. Besides, the optimal wealth in the region also increases. However, in the bad-states region, the optimal wealth becomes small with VaR constraint. Comparing with the optimization rule as in \cite{OIWS}, the good-states region shrinks and the bad-states region almost does not change. The optimal wealth in the bad-states region also increases. As such, we see that in the optimization problem with performance ratio, the manager sacrifices some gains in the good-states region to ensure the gains in the bad-states region.
\begin{figure}[htbp] 
	\centering
	\includegraphics[width=0.7\linewidth]{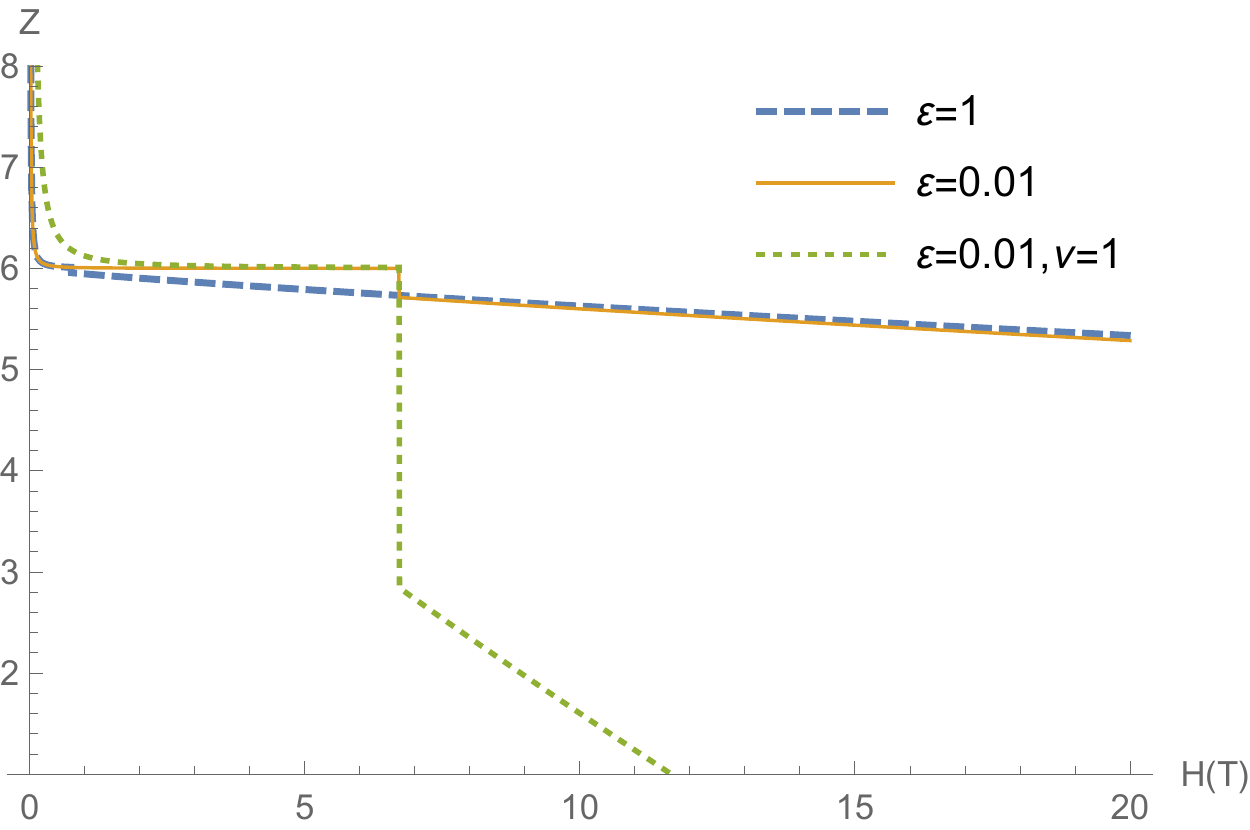}
	\caption{Graph of $ Z^* $ versus $ H(T) $ under different situations. }
	\label{comp}
\end{figure}

\vskip 10pt
\section{\bf Conclusion}
This paper studies the optimal management of DC pension fund under inflation risk. Different from many previous work, the optimization goal of the manager is the performance ratio defined as the reward function for over-performance divided by the penalty function for under-performance. Besides, to ensure the solvency of the fund, we consider VaR constraint at retirement time. The optimization problem is non-linear, non-self-financing and with additional constraint. We first transform the original problem into an equivalent self-financing model based on an auxiliary process. Then, by using martingale method, fractional programming method, Lagrange dual method and concavification  method, the optimal terminal wealth is obtained. To ensure the well-posedness of the problem, the reward function is concave. The optimal terminal wealth is presented for the concave and convex penalty function, separately.  For the convex penalty function, the manager has different risk aversions towards gains and losses in the linearized optimization problem  and there are in  total fourteen different cases. For the concave penalty function, the manager is risk seeking towards losses in the linearized optimization problem  and there are in total six cases, which is consistent with \cite{OIWS}.
\vskip 5pt
The optimal terminal wealth of the fund is obtained and the optimal investment strategies can be derived by replicating. Particularly, for the power reward and penalty functions, the explicit forms of the optimal investment strategies are obtained. There are two Lagrange multipliers in our paper, and we show the existence of them strictly in this paper. In some cases, the Lagrange multipliers may not exist and there is no optimal solution. In the end of this paper, we present sensitivity analysis to depict the economic behaviors of the pension manager.

\vskip 15pt
{\bf Acknowledgements.}
The authors acknowledge the support from the National Natural Science Foundation of China (Grant  No.11901574, No.11871036). The authors thank Dr. Litian Zhang  and
the members of the group of Mathematical Finance and Actuarial Science at the Department of Mathematical Sciences, Tsinghua University for their feedbacks and useful conversations.
\vskip 10pt

\appendix
\renewcommand{\theequation}{\thesection.\arabic{equation}}

\bibliographystyle{plainnat}
\bibliography{wpref}

\end{document}